\documentclass[11pt, letterpaper, onecolumn, oneside, final]{article}

\usepackage{url}
\usepackage{times}
\usepackage{xcolor}
\usepackage{xspace}
\usepackage{amsthm}
\usepackage{amsmath}
\usepackage{amssymb}
\usepackage{enumitem}
\usepackage{graphicx}
\usepackage[noend]{algpseudocode}
\usepackage[font=footnotesize,labelfont=bf]{caption} 
\usepackage[left=1in,right=1in,top=1in,bottom=1in]{geometry} 
\usepackage[pdftex,bookmarks=true,bookmarksnumbered=true]{hyperref}

\newtheoremstyle{thm-compact}
{1ex} 
{1ex} 
{\itshape} 
{} 
{\bfseries} 
{.} 
{ } 
{} 
\theoremstyle{thm-compact}
\newtheorem{definition}{Definition}[section]

\newtheorem{lemma}[definition]{Lemma}
\newtheorem{theorem}[definition]{Theorem}

\newtheorem{claim}{Claim}[definition] 

\newcommand{\backoff}[1]{#1-\textsc{backoff}}
\newcommand{\interference}[1]{#1-\textsc{interference}}

\begin{document}


\title{
	{\bf\centering Robust and Optimal Contention Resolution\\ without Collision Detection}
}
\author{
	Yonggang Jiang \\
	\small Max Planck Institute for Informatics,\\
	\small Saarland Informatics Campus \\
	\footnotesize\texttt{yjiang@mpi-inf.mpg.de}
	\and
	Chaodong Zheng \\
	\small State Key Laboratory for Novel Software Technology,\\
	\small Nanjing University \\
	\footnotesize\texttt{chaodong@nju.edu.cn}
}
\date{
}
\maketitle
\thispagestyle{empty}

\begin{abstract}
We consider the classical contention resolution problem where nodes arrive over time, each with a message to send. In each synchronous slot, each node can send or remain idle. If in a slot one node sends alone, it succeeds; otherwise, if multiple nodes send simultaneously, messages collide and none succeeds. Nodes can differentiate collision and silence only if collision detection is available. Ideally, a contention resolution algorithm should satisfy three criteria: low time complexity (or high throughput); low energy complexity, meaning each node does not make too many broadcast attempts; strong robustness, meaning the algorithm can maintain good performance even if slots can be jammed.

Previous work has shown, with collision detection, there are ``perfect'' contention resolution algorithms satisfying all three criteria. On the other hand, without collision detection, it was not until 2020 that an algorithm was discovered which can achieve optimal time complexity and low energy cost, assuming there is no jamming. More recently, the trade-off between throughput and robustness was studied. However, an intriguing and important question remains unknown: without collision detection, are there robust algorithms achieving both low total time complexity and low per-node energy cost?

In this paper, we answer the above question affirmatively. Specifically, we develop a new randomized algorithm for robust contention resolution without collision detection. Lower bounds show that it has both optimal time and energy complexity. If all nodes start execution simultaneously, we design another algorithm that is even faster, with similar energy complexity as the first algorithm. The separation on time complexity suggests for robust contention resolution without collision detection, ``batch'' instances (nodes start simultaneously) are inherently easier than ``scattered'' ones (nodes arrive over time).
\end{abstract}

\clearpage
\pagestyle{plain}
\setcounter{page}{1}


\section{Introduction}\label{sec-intro}

In computer systems, there are many scenarios in which a collection of players contend to access a shared resource. For example, a set of processes try to access a shared file on a hard drive or a table in a database, a set of radio transceivers each with a packet try to send these packets over a wireless channel, a set of users each with a document try to print them out using a printer, etc. In these settings, usually the goal is to let each player successfully access the shared resource (at least) once. However, the challenge lies in the requirement that successful accesses must be \emph{mutual exclusive}: if two or more players make access attempts simultaneously, a collision occurs and all these attempts fail.

In distributed and parallel computing, the above problem is known as \emph{contention resolution}. In studying this problem, often the shared resource is modeled as a synchronous multiple-access communication channel, and each player is modeled as a node with a message that needs to be sent over this channel. (See, e.g., \cite{bender18,chang19,bender20,chen21}.) More specifically, the system proceeds in synchronous time slots. Each player joins the system at the beginning of some slot, but players do not have access to any global clock. In each slot, each player may try to send its message by broadcasting it on the shared communication channel, or remain idle. If in a slot only one player $u$ broadcasts, it succeeds in that slot and all players are informed of this success. Moreover, by the end of that slot, $u$ will halt and exit the system. On the other hand, if in a slot multiple players broadcast simultaneously, then all of them fail as messages collide with each other. In such case, the exact channel feedback depends on the availability of a \emph{collision detection} mechanism. Specifically, when collision detection is available, for each slot not containing a success, the channel will inform all players whether the slot is silent (i.e., no node broadcasts) or noisy (i.e., multiple nodes broadcast and a collision occurs). By contrast, without collision detection, the channel feedback for each slot is binary: either the slot contains a success, or not.

\bigskip\noindent\textbf{Performance metrics.} To evaluate the performance of contention resolution algorithms, often the following three metrics are considered:
\begin{itemize}[topsep=2pt,partopsep=0pt,itemsep=0pt]
	\item \textsc{Throughput (Time complexity)}. Achieving a high throughput is usually the primary objective of any contention resolution algorithm. Although the exact definition of throughput depends on the specific model being considered, intuitively it captures the number of messages the algorithm can successfully transmit within a given period of time, representing how efficient can the algorithm process incoming requests. In synchronous systems, since each slot contains at most one successful transmission, constant throughput is the best result any contention resolution algorithm can hope for; that is, generating $n$ successes within $\Theta(n)$ slots.

	\item \textsc{Number of access attempts (Energy complexity)}. Due to the existence of collisions, it is likely that each player has to make multiple access attempts before succeeding. However, in many scenarios, there is a certain cost associated with each attempt. For example, systems calls must be invoked when processes are trying to access shared files, and system calls are ``expensive'' operations for operating systems~\cite{tanenbaum14}. Another prominent example is radio networks. Modern lightweight wireless devices are often battery powered, and the energy cost of emitting or receiving radio signals might be large when compared with computation~\cite{polastre05}. Therefore, contention resolution algorithms also try to minimize participating players' number of access attempts. In studying contention resolution, especially in the the context of radio networks, researches often call the maximum number of access attempts any node may make as the \emph{energy complexity} of the algorithm. (See, e.g., \cite{berenbrink09,chang18,bender18}.) In this paper, we also adopt this terminology.

	\item \textsc{Robustness}. Collision is not the only possible reason a transmission attempt fails. Participating nodes, or sometimes even the communication channel itself (e.g., when the channel models a shared printer), could suffer hardware or software errors. Even without internal failures, external interference might exist. For example, a wireless link may be affected by electromagnetic noise, a server may be the victim of denial-of-service attacks, etc. In studying contention resolution, these interference are often modeled by \emph{jamming}. (See, e.g., \cite{bender18,chang19,chen21}.) Formally, if a slot is jammed, a collision occurs in that slot, regardless of the actual number of broadcasting nodes. Robust contention resolution algorithms are preferable, sometimes even necessary. However, as we detail below, jamming could affect the complexity of the problem.
\end{itemize}

\bigskip\noindent\textbf{Existing results.} Given its importance and wide applicability, it is not surprising that contention resolution is extensively studied. Nonetheless, many important breakthroughs are only made recently. In particular, it was not until 2018 that a ``perfect'' contention resolution algorithm was proposed which simultaneously achieves high throughput and low energy complexity, even when adversarial jamming is present~\cite{bender18}. In particular, assuming $n$ nodes are injected over time by an adversary, and $d$ slots could be jammed by the adversary, Bender, Fineman, Gilbert, and Young proposed an algorithm called \textsc{Re-Backoff} guaranteeing: (a) constant throughput in expectation; and (b) each node makes $O(\log^2(n+d))$ access attempts in expectation. In 2019, Chang, Jin, and Pettie~\cite{chang19} proposed another algorithm that can achieve similar guarantees, by using a multiplicative weight update scheme. Compared with \cite{bender18}, this new algorithm is simpler and achieves higher throughput (the asymptotic throughput of these two algorithms are identical).

Both of the above two algorithms rely on the availability of collision detection. Generally speaking, if a node observers a noisy slot or many collisions within an interval, then the node can infer that the channel is congested and it should ``backoff'' by not sending its message. By contrast, if a node observers a silent slot or many empty slots within an interval, then the node should ``backon'' and send more frequently. Fundamentally, collision detection explicitly reveals why a slot fails, thus allowing nodes to act accordingly in the following slots. As a result, without collision detection, contention resolution becomes harder. Indeed, for a very long time, it is unclear whether constant throughput is possible if collision detection is not available, even without any interference.\footnote{Nonetheless, if all nodes start simultaneously instead of being injected over time, then algorithms achieving constant throughput without collision detection were already known for a while. See, e.g., \cite{mosteiro11}.}

Two years ago, Bender, Kopelowitz, Kuszmaul, and Pettie~\cite{bender20} resolved this long standing problem by providing an algorithm that achieves constant throughput without collision detection, even when the nodes are injected adversarially. They also proved a lower bound showing constant throughput is impossible when jamming is present. This demonstrates a fundamental separation between contention resolution with and without collision detection. More recently, Chen, Jiang, and Zheng~\cite{chen21} extended \cite{bender20} by taking adversarial interference into consideration. In particular, for any level of jamming ranging from none to constant fraction, they prove an upper bound on the best possible throughput, along with an algorithm attaining that bound. Unfortunately, there is still one important piece missing in the big picture. Although the algorithms in \cite{chen21} are able to achieve optimal throughout for any given level of jamming, their energy complexity could be high. In particular, assuming the adversary jams $d$ slots, there are nodes who might have to make $\Omega(\sqrt{d})$ attempts before succeeding. By contrast, when collision detection is available, achieving similar guarantees only requires $O(\textrm{poly-log}(n+d))$ attempts. Therefore, the intriguing question is: without collision detection, are there ``perfect'' contention resolution algorithms? In other words, when collision detection is not available and jamming is present, are there contention resolution algorithms that can achieve best possible throughout, while maintaining low per-node energy cost?

\bigskip\noindent\textbf{Contribution and new results.} In this paper, we answer the above question affirmatively. We design two new randomized algorithms for robust contention resolution, when collision detection is not available. One algorithm works for the case where $n$ nodes are injected over time by an adversary, which we refer to as the \emph{dynamic case}; while the other one works for the easier scenario where all $n$ nodes start simultaneously, which we refer to as the \emph{static case}. The time complexity---i.e., the number of slots required to let all nodes succeed---of the dynamic case algorithm is $O(n\log n+d)$, which is optimal given the lower bound proved in~\cite{chen21}. By contrast, the time complexity of the static case algorithm is $O(n+d)$, which clearly is also optimal. This separation suggests, for robust contention resolution without collision detection, ``batch'' instances (i.e., the static case) are inherently easier than ``scattered'' ones (i.e., the dynamic case). As for energy complexity, both algorithms ensure the number of access attempts any node made is $O(\log^2n+\log^2d)$. This is optimal regarding robustness: we show any algorithm achieving the optimal time complexity incurs $\Omega(\log^2d)$ energy cost per-node.
Lastly, we note that another advantage of our algorithms lies in simplicity. For example, the core component of the dynamic case algorithm can be described in one sentence.

Before stating the results formally, we clarify some additional model details and assumptions. We often call the adversary Eve, and she is adaptive. Before execution starts, Eve decides a value $n$, meaning $n$ nodes will be injected into the system. These nodes do not know the value of $n$. Eve also has a jamming budget $d$, meaning she could jam up to $d$ slots. In the static case, Eve injects (i.e., activates) all nodes at the beginning of slot one; whereas in the dynamic case, she can inject nodes in an arbitrary fashion. Therefore, in the dynamic case, there might be slots in which there are no active nodes, yet some nodes are still not injected by Eve. (Recall a node halts and leaves once its message is successfully sent.) We say a slot is an \emph{active slot} if in that slot at least one node is active. The adaptivity of Eve is reflected by the assumption that, at the beginning of each slot, Eve is given the past behavior of all nodes, and she can use this information to determine her behavior in the current slot. (Specifically, whether to inject any new nodes and whether to jam this slot.) However, Eve does not know nodes' behavior in the current slot.

The following definition introduces \emph{$(f,g)$-time-cost} and \emph{$(f,g)$-energy-cost}. It formalizes the notation of throughput and energy complexity, and simplifies later presentation.

\begin{definition}[\textbf{Throughput and Energy Complexity}]\label{def:throughput-and-energy}
Let $\mathcal{A}$ be the algorithm each node runs after its activation. Let $f,g:\mathbb{N}^+\rightarrow\mathbb{R}^+$ be two functions.
\begin{itemize}[nosep]
	\item A slot is an \emph{active slot} if at least one node is active in that slot.
	\item Algorithm $\mathcal{A}$ achieves \emph{$(f,g)$-time-cost} if there exists a constant $C$ such that for any integer $n,d\ge 1$ and any adaptive adversary that injects $n$ nodes and has jamming budget $d$, the total number of active slots is at most $C\cdot\left(f(n) + g(d)\right)$, with high probability in $n+d$.\footnote{An event happens with high probability (w.h.p.) in some parameter $\lambda$ if it happens with probability at least $1-1/\lambda^{\beta}$, for some desirable constant ${\beta}\geq 1$. Here $\lambda=n+d$ (instead of just $n$) since we want the failure probability to go down as the time cost grows up (time cost depends on both $n$ and $d$).}
	\item Algorithm $\mathcal{A}$ achieves \emph{$(f,g)$-energy-cost} if there exists a constant $C$ such that for any integer $n,d\ge 1$ and any adaptive adversary that injects $n$ nodes and has jamming budget $d$, the maximum number of broadcasting attempts any node made is at most $C\cdot\left(f(n) + g(d)\right)$, with high probability in $n+d$.
\end{itemize}
\end{definition}

The following two theorems state our algorithmic results, notice the difference on the time complexity.

\begin{theorem}[\textbf{Dynamic Case Upper Bound}]\label{thm:dynamic-single-channel}
There exists an algorithm achieving $(n\log{n},d)$-time-cost and $(\log^2{n},\log^2{d})$-energy-cost.
That is, this algorithm generates $n$ successes within $O(n\log{n}+d)$ active slots, and each node makes at most $O(\log^2{n}+\log^2{d})$ access attempts, with high probability in $n+d$.
\end{theorem}

\begin{theorem}[\textbf{Static Case Upper Bound}]\label{thm:static-single-channel}
Suppose all nodes start execution simultaneously, then there exists an algorithm achieving $(n,d)$-time-cost and $(\log^2{n},\log^2{d})$-energy-cost.
That is, this algorithm generates $n$ successes within $O(n+d)$ active slots, and each node makes at most $O(\log^2{n}+\log^2{d})$ access attempts, with high probability in $n+d$.
\end{theorem}

The following two theorems state our lower bound results.

\begin{theorem}[\textbf{Dynamic Case Lower Bound}]\label{thm:lower-bound-dynamic}
Suppose an algorithm achieves $(f_t,g_t)$-time-cost and $(f_e,g_e)$-energy-cost with $g_t(d)=d$, then $f_t(n)=\Omega(n\log n)$, $f_e(n)=\Omega(\log^2n)$, and $g_e(d)=\Omega(\log^2d)$.
\end{theorem}

\begin{theorem}[\textbf{Static Case Lower Bound}]\label{thm:lower-bound-static}
Suppose all nodes start execution simultaneously and an algorithm achieves $(f_t,g_t)$-time-cost and $(f_e,g_e)$-energy-cost with $f_t(n)=n$ and $g_t(d)=d$, then $f_e(n)=\Omega(\log\log{n})$ and $g_e(d)=\Omega(\log^2d)$.
\end{theorem}

The above lower bounds show that our dynamic case algorithm is optimal on both time complexity and energy complexity. On the other hand, our static case algorithm achieves optimal time complexity, as well as optimal energy complexity regarding the dependency on $d$; but it misses a bit on the energy complexity regarding the dependency on $n$. Therefore, when $d$ is small compared with $n$, in the static case, there might exist an algorithm with better energy complexity: this remains to be an open problem.

We conclude this part by noting that our lower bounds hold even for a weaker oblivious adversary.

\bigskip\noindent\textbf{Additional related work.} Perhaps the most classical algorithm to resolve contention is \emph{binary exponential backoff}. One standard implementation of it is to let each node send its message with probability $1/i$ in the $i$-th slot since the node's activation. Despite binary exponential backoff is widely used in practice (e.g., Ethernet and WiFi networks~\cite{kurose17}, concurrency control in operating systems and database management systems~\cite{tanenbaum14,ramakrishnan02}), it is long known that this scheme cannot always achieve optimal throughput~\cite{aldous1987,hastad87,bender05}. Therefore, many variants are proposed and analyzed, such as quadratic backoff~\cite{hastad87}, saw-tooth backoff, loglog-iterated backoff~\cite{bender05}. Our algorithms also utilize variants of binary exponential backoff.

As we have mentioned previously, the availability of collision detection can greatly affect the performance of contention resolution algorithms. Around 2010, a series of elegant results were published (see, e.g., \cite{awerbuch08,richa10}), demonstrating how constant throughput could be attained by using collision detection, even when jamming is present, in the context of wireless networks. Over the last few years, similar results were also obtained in the standard multiple-access communication channel model~\cite{bender18,chang19}. However, it was not until 2020 that a algorithm was developed which can achieve constant throughput without collision detection~\cite{bender20}. This paper also focuses on the more challenging scenario in which collision detection is not available.

Besides throughput, the number of channel accesses before succeeding is another key performance metric. This is especially relevant in radio networks since energy consumption of radio transceivers often dominate the total energy expenditure of wireless devices~\cite{polastre05}, and contention resolution is closely related to many fundamental communication primitives (such as gossiping, broadcasting) in radio networks~\cite{berenbrink09,chang18,chang19leader}. Thus, the number of channel accesses is also referred as energy complexity in the literature. Existing contention resolution algorithms achieving good throughput usually have energy complexity that are poly-logarithmic in the number of participating nodes, though \cite{bender16} shows in fact $O(\log(\log^*(n)))$ accesses is sufficient in expectation.\footnote{In fact, this $O(\log(\log^*(n)))$ bound counts both the number of ``send'' and ``listen'', where ``listen'' means obtaining channel feedback for one slot. In this paper, we only consider ``sending complexity'' and assume channel feedback is provided for free. This assumption is used in many works studying contention resolution. On the other hand, the assumption that ``listen'' is not free is usually made in the context of radio networks.} However, when collision detection is not available and jamming is present, to the best of our knowledge, no existing work can achieve good throughput while maintaining low energy complexity. Our paper addresses this open problem.

Although this paper and many previous work assume worst-case (i.e., adversarial) arrival pattern, another major line of research assumes the arrivals of nodes follow some statistical pattern. In those works, often the main objective is to analyze the maximum stable packet arrival rate for various contention resolution algorithms. (See, e.g., \cite{aldous1987,goodman88,goldberg00}.)

It is also worth noting, if a single success is sufficient (instead of requiring every node to succeed once), then contention resolution degrades to another classical distributed computing problem: leader election. Leader election is used implicitly in many contention resolution algorithms. A classical result by Willard~\cite{willard86} shows a tight time bound of $\Theta(\log\log{n})$ for leader election, assuming $n$ nodes start simultaneously and collision detection is available. More recent results consider more diverse settings (see, e.g., \cite{chlebus05,dolev09,chang19leader}). Interestingly, it seems our lower bounds could also apply to the leader election problem, as in deriving them we consider the time and energy required to generate the first success.

Finally, we stress that contention resolution is an extensively studied problem, only most relevant results are briefly discussed here, and many interesting results are not included (e.g., there are papers focusing on deterministic algorithms~\cite{marco19}, and papers considering messages with delivery deadlines~\cite{agrawal20}).

\bigskip\noindent\textbf{Paper outline.} In the next section, we give an overview on the design and analysis of the two new contention resolution algorithms, as well as the key ideas we exploit in proving the lower bounds. Then, in Section~\ref{sec:dynamic-alg} and Section~\ref{sec:static-alg}, we introduce the two algorithms in detail. For each of these two sections, we will first give a more through introduction on the intuition of the algorithm, then present the pseudocode, and finally proceed to the analysis. Next, in Section~\ref{sec:lower-bounds}, we prove the lower bounds on time complexity and energy complexity. We conclude this paper with Section~\ref{sec:future-work}, where we briefly discuss potential future work directions.

\section{Technical Overview}\label{sec:overview}

\textbf{Contention.} When designing efficient and robust contention resolution algorithms, the key is to maintain a proper \emph{contention} on the communication channel throughout the execution. Specifically, the contention of a slot on a channel is defined to the sum of the broadcasting probabilities of all active nodes. By definition, the contention of a slot denotes the expected total energy expenditure of nodes in that slot. On the other hand, the contention of a slot also indicates the likelihood of a slot being a successful slot. In particular, the follow lemma holds, in which $n$ denotes the number of active nodes in a slot, $p_i$ denotes the broadcast probability of node $i$, and $p$ denotes the contention of that slot. (Its proof is provided in the appendix.)

\begin{lemma}[Extension of Lemma 4 of \cite{bender20}]\label{lem:contention-to-succeess-probability}
Let $n\in\mathbb{N}^+$, and let $X_i\in\{0,1\}$ be an indicator random variable with $\Pr[X_i=1]=p_i$ for all $i\in[n]=\{1,2,\cdots,n\}$. Let $p=\sum_{i=1}^np_i$.
\begin{itemize}
	\item $\Pr\left[\left(\sum_{i=1}^{n}X_i\right)=1\right]\ge\min\left(4^{-p},\frac{p}{4}\right)\textnormal{~and~}\Pr\left[\left(\sum_{i=1}^{n}X_i\right)=0\right]\ge 4^{-p}\textnormal{, when all }p_i\in[0,1/2]$
	\item $\Pr\left[\left(\sum_{i=1}^{n}X_i\right)=1\right]\le p\cdot e^{-p+1}$
\end{itemize}
\end{lemma}

Some important implications of the above lemma are: (a) if the contention of a slot is some constant, then the probability that this slot generates a success is some constant; (b) if the contention of a slot is sufficiently large in $\Omega(\log{\lambda})$ where $\lambda>1$, then with high probability in $\lambda$ this slot will not generate a success; and (c) if the contention of a slot is sufficiently small in $O(\log\lambda)$ where $\lambda>1$, then the success probability in this slot is at least $1/\lambda^\beta$ for some constant $0<\beta<1$. By combining (b) and (c), we know that $\Theta(\log\lambda)$ is the ``right'' contention if we want to generate a success in $\Theta(\lambda)$ slots.

\bigskip\noindent\textbf{Exponential backoff.} Recall a standard way to implement binary exponential backoff is to let each node broadcast with probability $1/i$ in the $i$-th slot since its arrival. Consider the scenario in which $n$ nodes start running binary exponential backoff simultaneously, observe how the contention evolves. In the first $n$ slot, the contention is $\Omega(1)$ and limited successes will occur; particularly, at least a constant fraction of all $n$ nodes will remain active by the end of slot $n$. Then, from slot $n+1$ to, say, slot $10n$, the contention will remain to be a constant. By Lemma \ref{lem:contention-to-succeess-probability}, we know in expectation a success will occur every some constant slots. Thus, by the end of slot $10n$, at most some constant fraction of all $n$ nodes will remain active. Lastly, from slot $10n+1$, the contention will continue to drop, and eventually reach $o(1)$, which is too small for successful transmissions to occur frequently. In short, the first $\Theta(n)$ slots of binary exponential backoff are efficient in that $\Theta(n)$ successes are likely to occur, and each node's energy cost is $O(\log{n})$. Nonetheless, beyond this interval, throughput will drop. In this paper, we will use both standard binary exponential backoff and its variants. To simplify presentation, we define the following generalized backoff pattern.

\begin{definition}(\backoff{$h$})\label{def:backoff}
Let $h:\mathbb{N}^+\rightarrow\mathbb{R}^+$ be a function. We say a node runs \emph{\backoff{$h$}} from slot $s$, if for any $x\in\mathbb{N}^+$, the node sends its message with probability $\min\{1,h(x)\}$ in slot $s+x-1$.
\end{definition}

For example, \backoff{$1/x$} is the standard binary exponential backoff.

\bigskip\noindent\textbf{The two-channel model.} When designing and analyzing our algorithms, it is often more convenient to assume there are \emph{two} channels. Here, we introduce the two-channel model, and defer the discussion on how to covert a two-channel algorithm to a single-channel algorithm to later parts of the paper.

Simply put, the two channels act like two independent multiple-access communication channels: in each slot, for each of the two channels, each active node and the adversary can independently decide its behavior on this channel: broadcast its message, jam the channel, or remain idle. In each slot, both channels give feedback to the active nodes. Time complexity naturally extends to the two-channel model. As for energy complexity, if in a slot a node sends on both channels (respectively, Eve jams both channels), then the energy expenditure of this node (respectively, the jamming budget Eve spends) is two.

\bigskip\noindent\textbf{Overview of algorithm for the dynamic scenario.} Recall we have argued above that the first $\Theta(n)$ slots of binary exponential backoff are efficient. In fact, in previous works that achieve constant throughput and do not rely on collision detection, a core idea is to ``repeat the efficient part of the exponential backoff process''~\cite{bender20,chen21}. However, when jamming is present, the energy complexity of this scheme grows quickly. Specifically, for each repetition of exponential backoff, Eve can jam all the first $\Theta(n)$ slots and then allow a success to occur so that the next repetition starts. If this process is repeated $n/2$ times, then the energy consumption of Eve is $\Theta(n^2)$, and the energy consumption of each of the remaining $n/2$ nodes is $\Omega(n\log{n})$. That is, if Eve can jam $d=\Theta(n^2)$ slots, then the energy consumption of a node might reach $\Omega(\sqrt{d}\log{n})$. To overcome this difficulty, we let each node simply run one instance of backoff from the slot it joins the system to the slot it succeeds.

The above modification enforces good energy complexity: assuming all nodes start simultaneously and the adversary jams the first $d$ slots, the cost of each node in the first $d$ slots would be only $\Theta(\log{d})$. However, this simple tweak results in sub-optimal time complexity: if Eve jams the channel sufficiently long and then stop, it would take remaining nodes too long too send out their messages.

This is where the second channel comes into play. For each node, upon arrival, it will run a backoff procedure with more ``aggressive'' sending probabilities on channel two. In particular, \backoff{$(c\log{x})/x$} is a good choice, where $c$ is a large constant. Together with channel one, these two backoff procedures guarantee good time complexity regardless of the jamming pattern. In particular, Eve can jam sufficiently long to let the contention of both channels become $o(1)$, but in such case channel two can still generate successes sufficiently often, at least in an amortized sense, giving $O(n\log{n}+d)$ total time complexity.

The above scheme works even if nodes are injected over time by the adversary (though the analysis becomes harder), so the only remaining issue is to let it work in the single-channel model. Notice that if nodes can access (global) slot indices, then a simple solution exists: let all odd slots correspond to the first channel, and all even slots correspond to the second channel. Unfortunately, we do not assume nodes have access to such information. Instead, we use a ``synchronization procedure'' inspired by \cite{bender20}. For each node $u$, upon arrival, it will first run the synchronization procedure; when the synchronization procedure ends, all nodes that are in the system reach agreement on the parity of slots, so $u$ can start running the two-channel algorithm (and other existing nodes resume running the two-channel algorithm).

\bigskip\noindent\textbf{Overview of algorithm for the static scenario.} In this scenario, we gain the advantage that all nodes start simultaneously. This allows us to easily convert any two-channel algorithm into a corresponding single-channel algorithm, since all nodes agree on slot indices. Nonetheless, in this setting, we are also targeting a better time complexity: $O(n+d)$ slots. Therefore, the above two-channel algorithm is inadequate.

Recall that when collision detection is not available and jamming is not present, to achieve constant throughput, previous works rely on repeating the ``efficient part'' of the binary exponential backoff procedure. Specifically, imagine there are $n$ nodes running \backoff{$1/x$} on one channel---called the data channel, and these nodes also run \backoff{$(c\log{x})/x$} on another channel---called the control channel. Then it can be shown, the first success on the control channel will happen in slot $\Theta(n)$. By then, remaining nodes will restart \backoff{$1/x$} on the data channel, and restart \backoff{$(c\log{x})/x$} on the control channel.

However, restarting two backoff procedures from scratch is not necessary. Indeed, if $n$ nodes run \backoff{$1/x$} and \backoff{$(c\log{x})/x$}, then in the first $n$ slots, the contention on the two channels are at least $1$ and $c\log{n}$ respectively, implying not many successes will occur in these $n$ slots. Hence, in our static algorithm, each node maintains a variable $\ell$ to control its sending probability: in each slot, each node sends on the data channel with probability $1/\ell$ and sends on the control channel with probability $(c\log{\ell})/\ell$. Initially $\ell=1$, then in each of the following slots, if the control channel does not generate a success, the node increases $\ell$ by one; otherwise, if the control channel does generate a success, the node halves the value of $\ell$. Compared with the scheme discussed in the last paragraph (which is setting $\ell=1$ after each control channel success), this new scheme only doubles the sending probability upon seeing a control channel success, efficiently maintaining the contention of the data channel within a desirable interval.

Unfortunately, this doubling scheme still suffers poor energy efficiency if Eve focus on jamming slots in which the contention of the data channel is $\Theta(1)$. Nevertheless, a critical observation is, this doubling scheme only suffers poor energy efficiency if Eve jams at least $\Omega(n\log{n})$ slots, which implies $d=\Omega(n\log{n})$. But now that $d=\Omega(n\log{n})$, the time complexity we are targeting---which is $O(n+d)$---is dominated by $d$. In other words, when $d=\Omega(n\log{n})$, after running the above doubling scheme for a while, the remaining nodes can switch to running another algorithm. This phase two algorithm can afford $\Theta(d)$ running time (which might be $\omega(n)$), in exchange for good energy efficiency.

So the condition for the transition from phase one to phase two is crucial. In our final algorithm, whenever a control channel success occurs, the value of $\ell$ (after halving) is recorded. Nodes use a variable $m$ to maintain the minimum of these recorded values. By the end of each slot $t$, if $m\leq t/\log{t}$, remaining nodes will switch to phase two, which is to run a simple \backoff{$(c\log{x})/x$} from scratch on the control channel. In Section~\ref{sec:static-alg}, we will discuss the effectiveness of this transition condition in detail.

\bigskip\noindent\textbf{Lower bounds.} For any robust contention resolution algorithm that does not utilize collision detection, before the first success, nodes cannot differentiate the following cases: (a) they have small sending probabilities, creating a contention too low; (b) they have large sending probabilities, creating a contention too high; or (c) the contention is right but the adversary is jamming. This dilemma suggests, if nodes want to achieve a good time complexity (even for creating the first success), they have to account for the first possibility and send sufficiently often upon arrival. Exploiting this observation allows us to prove a key technical lemma which connects the energy complexity and the time complexity of contention resolution algorithms.

Notice that the static scenario has trivial time complexity lower bound $\Omega(n+d)$, since each slot generates at most one success and Eve can jam $d$ slots continuously to block any communication. As for the dynamic scenario, the $\Omega(n\log{n}+d)$ time complexity bound is a direct corollary of Theorem 1.3 of \cite{chen21}.

Lastly, to obtain the energy complexity lower bounds, we combine the time complexity lower bounds and the lemma just described.

\bigskip\noindent\textbf{Concentration inequalities.} We conclude this section by introducing two concentration inequalities that will be used frequently in the analysis. The first one is the so-called ``convenient'' Chernoff bound, which can be found in various textbooks on randomized algorithms (such as \cite{mitzenmacher17}); the second one is Lemma 3.4 from \cite{chen21}, which in turn relies on Lemma 3 of \cite{bender20}.

\begin{lemma}[Chernoff Bound]\label{lem:chernoff-bound}
Suppose $X_1,X_2,\cdots,X_N$ are $N$ independent 0-1 random variables such that $\Pr[X_i=1]=p_i$ for all $1\leq i\leq N$. Let $X=\sum_{i=1}^{N}X_i$, then we have:
\begin{itemize}
	\item $\Pr[X\geq(1+\delta)\mathbb{E}[X]]\leq\exp\left(-\frac{\delta^2\mathbb{E}[X]}{3}\right)\textnormal{ for any }0<\delta\leq 1$
	\item $\Pr[X\leq(1-\delta)\mathbb{E}[X]]\leq\exp\left(-\frac{\delta^2\mathbb{E}[X]}{2}\right)\textnormal{ for any }0<\delta<1$
	\item $\Pr[X\geq R]\leq 2^{-R}\textnormal{ for any }R\geq 6\mathbb{E}[X]$
\end{itemize}
\end{lemma}

\begin{lemma}[Lemma 3.4 of \cite{chen21}]\label{lem:cjz21}
Suppose $X_1,X_2,\cdots,X_N$ are $N$ (potentially dependent) random variables such that $\Pr[X_i=t~|~X_1=x_1, X_2=x_2,\cdots, X_{i-1}=x_{i-1}]\leq 1/t^{\Omega(1)}$ holds for any sufficiently large $t$, any $1\leq i\leq N$, and any values $x_1,x_2,\cdots,x_{i-1}$ of $X_1,X_2,\cdots,X_{i-1}$, then with high probability in $N$, we have $\sum_{i=1}^{N} X_i=O(N)$.
\end{lemma}

\section{The Dynamic Scenario}\label{sec:dynamic-alg}

In this section, we introduce our algorithm for the scenario where nodes are injected over time by the adversary. We will first give a more through discussion on the intuition of the algorithm in the two-channel model, then present the pseudocode, and finally proceed to the analysis. We will conclude this section by introducing a mechanism that allows the algorithm to work in the single-channel model.

Recall the high-level structure of the dynamic two-channel algorithm we introduced in Section~\ref{sec:overview}. To avoid the poor energy efficiency brought by ``repeating the efficient part of exponential backoff'', each node runs one instance of \backoff{$1/x$} continuously on the first channel upon arrival. However, a side effect of this mechanism is sub-optimal time complexity: if the adversary jams for a sufficiently long time and then stop disruption, it would take remaining nodes too long to generate enough successes, since their broadcasting probabilities are too low by then. To fix this problem, for each node, upon its arrival, we let it run one instance of \backoff{$(c\log{x})/x$} continuously on the second channel, where $c$ is a large constant.

To understand why the second channel helps reduce time complexity, consider the simple case where all $n$ nodes start simultaneously, and let us examine how the contention evolves in the above two-channel algorithm. In the first $n$ slots, the contention on both channels are $\Omega(1)$ and limited successes would occur, even if Eve does no jamming. This part is inherently inefficient but short, so overall it has no large impact on performance. Next, consider slots from $n+1$ to $\lambda_1n\log{n}$ where $\lambda_1$ is some large constant. If Eve does not jam a constant fraction of these slots on channel one, then the first channel alone would generate enough successes so that all nodes would halt by the end of slot $\lambda_1n\log{n}$ (see \cite{bender05}). In such case, the overall time complexity is $O(n\log{n}+d)$. If, on the contrary, Eve does jam at least constant fraction of slots from $n+1$ to $\lambda_1n\log{n}$ on channel one, then $d=\Omega(n\log{n})$. In such case, assume slot $s_d$ is the first slot after slot $\lambda_1n\log{n}$ that is not jammed by Eve, then we know $d=\Omega(s_d)$. Consider a node $u$ that is still active in slot $s_d$, its sending probability on the second channel is $(c\log{s_d})/s_d$. Notice that the contention on channel two in slot $s_d$ is at most $n\cdot(c\log{(\lambda_1n\log{n})})/(\lambda_1n\log{n})=O(1)$, as $s_d\geq\lambda_1n\log{n}$. Therefore, in slot $s_d$ and the $\Theta(s_d)$ slots following slot $s_d$, the probability that $u$ succeeds on channel two is at least $\Theta((\log{s_d})/s_d)$. In other words, in interval $[s_d,2s_d]$, if Eve does not jam a constant fraction of these slots on channel two, then $u$ will halt by the end of slot $2s_d=O(d)$, with high probability in $s_d=\Omega(n\log{n})$. A union bound would then imply this claim holds for every node that is still active in slot $s_d$. To sum up, informally speaking, if the two-channel algorithm is executed for a duration sufficiently large in $\Omega(n\log{n})$, and if Eve does not jam some constant fraction of these slots on at least one channel, then all nodes must have succeeded by the end of these slots. In other words, the time complexity of the two-channel algorithm is $O(n\log{n}+d)$. This would also imply the energy cost of each node is bounded by $O(\log^2{(n\log{n}+d)})$, which could also be expressed as $O(\log^2{n}+\log^2{d})$.

The above discussion intuitively illustrates the effectiveness of the two-channel algorithm, later in the analysis subsection we would show it provides similar guarantees even if the $n$ nodes are injected dynamically by the adversary. The high-level argument, which is a generalization of the above discussion, being: (a) the total number of slots in which the contention on the first channel is $\Omega(1)$ cannot be too large, as each node runs a backoff instance continuously; (b) during the time period in which the contention of the first channel is between $\Theta(1)$ and $\Theta(1/\log{n})$, if $\Theta(n\log{n})$ such slots are not jammed, all $n$ nodes would succeed; and (c) for slots where the contention of the first channel is $O(1/\log{n})$, Eve must have jammed sufficiently many slots previously to let such slots occur, thus though successful transmissions will not occur too frequently, the overall time complexity can still be bounded by an amortized analysis.

\subsection{Algorithm Description}\label{sec:dynamic-alg-description}

Assuming nodes can access two independent communication channels, the algorithm for the dynamic case is very simple: for each node, in the $i$-th slot since its activation, send with probability $1/i$ on channel one, and send with probability $(c\log{i})/i$ on channel two. Here, $c>0$ is some sufficiently large constant. In other words, the dynamic algorithm can be described using one sentence:

\addvspace{.5\baselineskip}
{\centering
\framebox[\textwidth][c]{
\parbox{.95\textwidth}{
\par\addvspace{.3\baselineskip}
\hypertarget{alg:dynamic}{\underline{\emph{Algorithm for each node in the dynamic scenario:}}}
\par\addvspace{.3\baselineskip}
From the arriving slot, run \backoff{$(1/x)$} on channel one, and \backoff{$(c\log{x})/x$} on channel two.
}}}
\par\addvspace{.5\baselineskip}

As mentioned earlier, we defer the discussion on how to convert the above two-channel algorithm to the single-channel setting to the last part of this section.

\subsection{Algorithm Analysis}\label{sec:dynamic-alg-analysis}

In this subsection, we formally analyze the performance of the dynamic algorithm in the two-channel setting. As mentioned earlier, we assume jamming on one channel for one slot costs one unit of energy, and Eve has a total energy budget of $d$. On the other hand, to facilitate the process that converts a two-channel algorithm to the single-channel setting, we extend Eve's ability on injecting nodes:

\begin{definition}[\interference{$h$}]\label{def:h-interference}
Let $h:\mathbb{N}^+\rightarrow\mathbb{R}^+$ be a function. In the two-channel model, we say the adversary has the ability of \emph{\interference{$h$}} if she can inject additional nodes called \emph{interference nodes}. For each interference node, on the first channel, the adversary can specify whether the node starts running \backoff{$h$} from the slot the node is injected or from the slot following the slot the node is injected; for each interference node, on the second channel, the node always starts running \backoff{$h$} from the slot it is injected. An interference node will halt and leave the system upon hearing a success on any channel.
\end{definition}

In analyzing the dynamic scenario two-channel algorithm, we assume the adversary has the ability of \interference{$(c\log{x})/x$}. In particular, she can inject up to $n$ interference nodes beside the $n$ normal nodes. In such setting, the definition for the throughput and the energy complexity of a two-channel algorithm, as well as the definition for active slot, need to be adjusted accordingly.

\begin{definition}[Throughput and Energy Complexity with \interference{$(c\log{x})/x$}]\label{def:throughput-and-energy-with-interference}
Let $\mathcal{A}$ be the two-channel algorithm each normal node runs after arriving. Let $f,g:\mathbb{N}^+\rightarrow\mathbb{R}^+$ be two functions.
\begin{itemize}[nosep]
	\item A slot is \emph{active} if either some interference node or some normal node is still active in that slot.
	\item Algorithm $\mathcal{A}$ achieves \emph{$(f,g)$-time-cost} if there exists a constant $C$ such that for any integer $n,d\ge 1$ and any adaptive adversary with jamming budget $d$ that injects $n$ normal nodes and up to $n$ interference nodes, the number of active slots is at most $C\cdot\left(f(n) + g(d)\right)$, with high probability in $n+d$.
	\item Algorithm $\mathcal{A}$ achieves \emph{$(f,g)$-energy-cost} if there exists a constant $C$ such that for any integer $n,d\ge 1$ and any adaptive adversary with jamming budget $d$ that injects $n$ normal nodes and up to $n$ interference nodes, the maximum number of broadcasting attempts a normal node or an interference node made is at most $C\cdot\left(f(n) + g(d)\right)$, with high probability in $n+d$.
\end{itemize}
\end{definition}

The following lemma reveals the fact that by allowing the adversary to \interference{$(c\log{x})/x$}, we can convert a two-channel algorithm to the single-channel setting, with guarantees on both time complexity and energy complexity. We defer the proof of it to the part where we discuss the conversion in detail.

\begin{lemma}\label{lem:two-to-single}
Suppose there is a two-channel algorithm achieving $(n\log{n}, d)$-time-cost and $(\log^2{n}, \log^2{d})$-energy-cost when the adversary can \interference{$(c\log{x})/x$} and inject $n$ interference nodes, then there is a single-channel algorithm achieving $(n\log{n}, d)$-time-cost and $(\log^2{n}, \log^2{d})$-energy-cost.
\end{lemma}

The reminder of this subsection focus on analyzing the runtime of the two-channel algorithm when the adversary can \interference{$(c\log{x})/x$}. Specifically, we intend to prove the following theorem.

\begin{theorem}\label{thm:dynamic-two-channel}
The two-channel algorithm achieves $(n\log{n},d)$-time-cost and $(\log^2{n}, \log^2{d})$-energy-cost, even if the adversary can \interference{$(c\log{x})/x$} and inject $n$ interference nodes.
\end{theorem}

To prove the above theorem, we divide active slots into two categories according to the contention on the second channel---ones that such contention reaches $0.5$ and ones that such contention is below $0.5$---and provide bounds for each of them. Throughout the analysis, for any slot, when calculating the contention on a channel, we sum up the broadcasting probabilities of both interference nodes and normal nodes.

\bigskip\noindent\textbf{Slots with large contention.} In this part, we count the number of active slots in which the contention of the second channel is at least $0.5$. To facilitate presentation, we introduce the notion of \emph{congest slots}.

\begin{definition}\label{def:dynamic-congest-slot}
We call a time slot a \emph{congest slot} if at least one of the following events happens:
\begin{itemize}[nosep]
	\item Channel one is jammed by the adversary;
	\item The contention created by normal nodes on channel one is at least $1$;
	\item There exists a (normal or interference) node on channel one that sends with probability at least $1/2$.
\end{itemize}
\end{definition}

With the above definition, we can further divide the slots we care (i.e., active slots in which the contention of the second channel is at least $0.5$) into the following three categories: (a) congest slots; (b) slots in which at least one interference node is active; and (c) remaining slots not in the above two categories.

Intuitively, due to the fact that each normal node runs \backoff{$1/x$} on the first channel upon arrival and that the jamming budget of the adversary is limited, an $O(n\log{n}+d)$ bound can be obtained on the number of slots that belong to the first category. As for the number of slots belonging to the last category, observe that in each such slot, the first channel is not jammed and no interference nodes are present; moreover, the contention of the first channel is both lower bounded (since the contention of the second channel is lower bounded and the contention of the two channels differ by a logarithmic factor) and upper bounded (due to the assumption that this slot is not a congest slot). Therefore, for category three slots, successes are likely to occur frequently, implying the number of such slots is limited.

Bounding the number of category two slots is a bit more involved. To that end, we introduce the following definition and lemma.

\begin{definition}\label{def:dynamic-interference-interval}
During an execution of the dynamic two-channel algorithm, define \emph{interference intervals} $I_1=[L_1,R_1],I_2=[L_2,R_2],...,I_k=[L_k,R_k]$ inductively as following. Define $R_0=0$. For any $i\in[k]$, $L_i$ is the first slot after $R_{i-1}$ where the adversary injects interference nodes; and $R_i$ is the first slot since $L_i$ in which there is a successful message transmission on any channel.
\end{definition}

\begin{lemma}\label{lem:interference-slots-count}
For any $i\in[k]$, denote the number of interference nodes injected during interference interval $I_i$ as $n'_i$, and denote the number of congest slots in interval $I_i$ as $d'_i$. Then for any positive integer $t>C^3(n'_i\log n'_i+d'_i)$ where $C$ is a sufficiently large constant, for any fixed $\ell_1,\ell_2,...,\ell_{i-1}$, we have $\Pr\left[~|I_i|=t~\mid~|I_1|=\ell_1,|I_2|=\ell_2,...,|I_{i-1}|=\ell_{i-1}~\right]\le1/t^{\Omega(1)}$.
\end{lemma}

\begin{proof}
Suppose $|I_i|=t$. Recall that each interference node runs \backoff{$(c\log x)/x$}, so during $I_i$ each interference node contributes at most $c\log^2t$ to the sum of the contention of the first channel. Thus, during $I_i$, the total contention created by interference nodes on the first channel is at most $cn'_i\log^2t$.

We now argue there are at most $t/C$ slots in which interference nodes contribute at least $(c\log{t})/C$ to the first channel's contention. Otherwise, during $I_i$, on the first channel, the total contention created by the interference nodes is at least:
$$\frac{t}{C}\cdot\frac{c\log t}{C}=\frac{t}{\log t}\cdot\frac{c\log^2 t}{C^2}\ge\frac{C^3n'_i}{1+3\log C}\cdot\frac{c\log^2 t}{C^2}=\frac{C}{1+3\log C}\cdot cn'_i\log^2 t$$
where the inequality is due to the assumption $t>C^3n'_i\log n'_i$ and the fact that $t/\log t$ is increasing on $t$. For sufficiently large $C$, this leads to a contradiction, as we have shown above during $I_i$ the total contention created by interference nodes on the first channel is at most $cn'_i\log^2t$.

Therefore, there are at most $t/C$ slots in which interference nodes contribute at least $(c\log{t})/C$ to the first channel's contention. On the other hand, the assumption $t>C^3(n'_i\log n'_i+d'_i)$ implies $t>Cd'_i$, which in turn implies there are at most $t/C$ congest slots. As a result, during $I_i$, at least $(1-2/C)t$ slots are not congest slots, and the contention created by interference nodes on the first channel is at most $(c\log t)/C$ for each such slot. Call these slots $G_i$. According to the definition of congest slots, we know for each slot in $G_i$, the contention on the first channel is at most $1+(c\log t)/C$.

Recall that at least one interference node is injected in the first slot of $I_i$, so the total contention of the first channel is at least $(c\log t)/t$ throughout $I_i$. Hence, for each slot in $G_i$, the total contention of the first channel is between $(c\log t)/t$ and $1+(c\log t)/C$, and all these slots are not congest slots. As a result, by Lemma~\ref{lem:contention-to-succeess-probability}, for each slot in $G_i$, for sufficiently large $C$ (particularly, $C\gg c$), the probability that a success will occur is at least $(c\log t)/t$. Notice that $|G_i|\geq (1-2/C)t$, and that (normal and interference) nodes make choices independently among slots in $G_i$, we can apply a Chernoff bound and conclude, conditioned on any fixed history up to the beginning of $I_i$, the probability that no success occurs within $t>C^3(n'_i\log n'_i+d'_i)$ slots in $I_i$ is at most $(1-(c\log t)/t)^{(1-2/C)t}=1/t^{\Omega(1)}$.
\end{proof}

We are now ready to state and prove the main technical lemma of this part.

\begin{lemma}\label{lem:dynamic-num-large-contention-slots}
The number of active slots in which the contention on the second channel is at least $0.5$ is at most $O(n\log{n}+d)$, with high probability in $n+d$.
\end{lemma}

\begin{proof}
As mentioned previously, we divide the active slots we care into three categories, and provide an upper bound on size for each of them.

Let $d'$ be number of congest slots, and we first bound $d'$. Since each normal node runs \backoff{$1/x$} on channel one, each of them contributes at most $\log{(2n)}$ to the total contention of channel one during the first $2n$ slots after arriving. Therefore, if we inspect the union of the first $2n$ slots of all normal nodes, there will exist at most $2n\log{(2n)}$ slots in which the normal nodes contribute at least $0.5$ to the contention of channel one in that slot. On the other hand, for each slot after $2n$ slots since the last normal node is injected, on channel one, the contention created by the normal nodes is at most $n\cdot(1/(2n))=0.5$. Therefore, the number of slots in which the normal nodes contribute at least $1$ to the contention of the first channel is at most $2n\log{(2n)}$. Recall that the adversary can jam the first channel for at most $d$ slots, and that each (normal or interference) node sends with probability at least $1/2$ only in the first $O(1)$ slots since its arrival, we can conclude the number of congest slots $d'$ is at most $d+2n\log(2n)+n\cdot O(1)=O(d+n\log{n})$.

Next, we bound the number of slots with interference nodes. By the definition of interference interval (i.e., Definition~\ref{def:dynamic-interference-interval}), this is at most $\sum_{i=1}^k|I_i|$. For any $i\in[k]$, let $\ell_i$ be a random variable that equals to $|I_i|$ if $|I_i|>C(n'_i\log n'_i+d'_i)$ and equal to $0$ otherwise. Here, $n'_i$ and $d'_i$ are defined in the statement of Lemma~\ref{lem:interference-slots-count}. Let $\ell_i=0$ for any $k<i\le n+d$. According to Lemma~\ref{lem:interference-slots-count} and Lemma~\ref{lem:cjz21}, we have $\sum_{i=1}^{n+d}\ell_i=O(n+d)$, with high probability in $n+d$. Therefore, $\sum_{i=1}^k|I_i|\le\sum_{i=1}^k(\ell_i+C(n'_i\log n'_i+d'_i))\le O(n+d)+\sum_{i=1}^kC(n'_i\log n'_i+d'_i)=O(n\log n+d+d')=O(n\log{n}+d)$, w.h.p.\ in $n+d$.

For each active slot that is not a congest slot and has no interference node, the contention on the first channel is at most $1$. Let $G$ denote the set containing these slots. To bound $|G|$, we first bound the success probability of each slot in $G$, by giving a lower bound for channel one's contention. (We already have an upper bound for channel one's contention, which is $1$.) Consider one such slot, assume there are $m$ normal nodes. Assume the $i$-th node sends on channel two with probability $(c\log t_i)/t_i$, then the contention on channel two is $\sum_{i=1}^m{(c\log t_i)/t_i}\ge0.5$. (Recall that in this part we are considering ``slots with large contention'', so $\sum_{i=1}^m{(c\log t_i)/t_i}\ge0.5$.) Since the $i$-th node sends on channel one with probability $1/t_i$, the contention on channel one is at least $\sum_{i=1}^m{1/t_i}\ge1/(2c\log t)$, where $t$ is the maximum among all $t_i$. Due to the discussion in the above two paragraphs, we know with high probability in $n+d$, $t$ is at most $O(n\log n+d)+|G|$. According to Lemma~\ref{lem:contention-to-succeess-probability}, for each slot in $G$, a success will occur on channel one with probability $\Omega(1/\log(n\log{n}+d+|G|))$. Apply a Chernoff bound, we know when $|G|=\Theta(n\log n+d)$, there are $\Omega(|G|/\log(n\log{n}+d+|G|))=\Omega(n)$ successes in $G$, with high probability in $n+d$. In other words, $|G|=O(n\log n+d)$, with high probability in $n+d$.
\end{proof}

\bigskip\noindent\textbf{Slots with small contention.} In this part, we count the number of active slots in which the contention of the second channel is less than $0.5$.

Throughout this part, we treat the slots in which the contention on channel two reaches $0.5$ as jammed slots. Formally, we say an active slot is a \emph{busy slot} if in that slot at least one channel is jammed by the adversary or in that slot the contention on channel two reaches $0.5$. Assume there are $d'$ busy slots throughout the entire execution, due to Lemma~\ref{lem:dynamic-num-large-contention-slots}, we know $d'=O(n\log{n}+d)$.

To facilitate presentation, we introduce the notation of \emph{complete intervals}.

\begin{definition}\label{def:dynamic-complete-interval}
During an execution, the \emph{active interval} of a (normal or interference) node is the interval between the node's arriving and leaving (both inclusive). We define \emph{complete intervals} $I_1=[L_1,R_1],I_2=[L_2,R_2],...,I_k=[L_k,R_k]$ inductively as following. Define $R_0=0$. For any $i\in[k]$, let $O_i$ be the first active slot after $R_{i-1}$. Then, $I_i$ is defined to be the union of the active intervals that intersect with $O_i$.
\end{definition}

See Figure~\ref{fig:dynamic-complete-interval} for an example. By the definition above, all active slots of an execution are contained within the union of all complete intervals, so $\sum_{i=1}^k |I_i|$ is an upper bound on total number of active slots. However, bounding $\sum_{i=1}^k |I_i|$ differs from the proof of Lemma~\ref{lem:interference-slots-count}. In particular, interference intervals do not intersect with each other, so in the proof of Lemma~\ref{lem:interference-slots-count}, we can bound the length of each interference interval, conditioned on any previous execution history. By contrast, complete intervals may overlap.

\begin{figure}[!ht]
\centering
\includegraphics[scale=0.9]{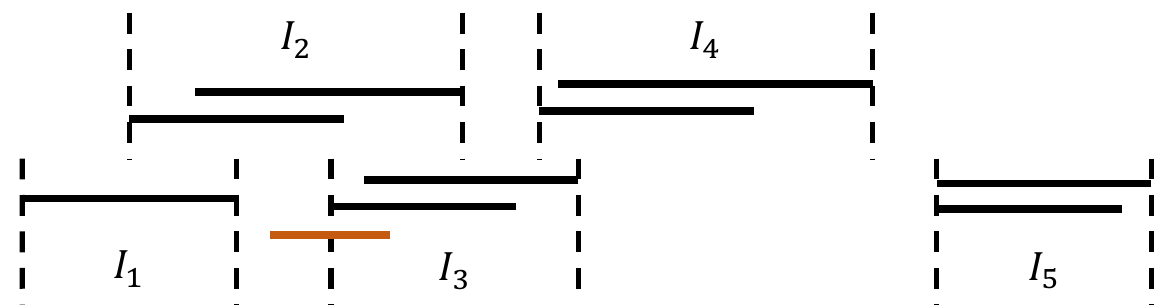}
\setcaptionwidth{0.85\textwidth}
\caption{An example of complete intervals, where each solid horizontal line denotes the active interval of a node. Complete intervals with the same parity of index do not overlap with each other. Notice that not every active interval is necessarily included in some complete interval. For example, the orange line in the figure.}\label{fig:dynamic-complete-interval}
\end{figure}

Nonetheless, an important observation about complete intervals is that $I_1,I_3,I_5,...$ are disjoint, so are $I_2,I_4,I_6,...$. This can be verified easily: any node with active interval intersecting a slot after $R_{i}$ must not intersect slot $R_{i-1}+1$, otherwise it should be included in interval $I_i$ and $R_i$ should be larger, which is a contradiction. As a result, we can bound $|I_1|+|I_3|+|I_5|+...$ and $|I_2|+|I_4|+...$ separately.

\begin{lemma}\label{lem:dynamic-complete-intervals}
For any $i\in[k]$, denote the number of injected (normal or interference) nodes during complete interval $I_i$ as $n_i$, and denote the number of busy slots in $I_i$ as $d'_i$. Then for any positive integer $t>C(n_i+d'_i)$ where $C$ is a sufficiently large constant, for any odd $i$ and any fixed $\ell_1,\ell_3,...,\ell_{i-2}$, we have $\Pr\left[~|I_i|=t ~\mid~ |I_1|=\ell_1,|I_3|=\ell_3,...,|I_{i-2}|=\ell_{i-2}~\right]\le1/t^{\Omega(1)}$. Similarly, for any even $i$ and any fixed $\ell_2,\ell_4,...,\ell_{i-2}$, we have $\Pr\left[~|I_i|=t ~\mid~ |I_2|=\ell_2,|I_4|=\ell_4,...,|I_{i-2}|=\ell_{i-2}~\right]\le1/t^{\Omega(1)}$.
\end{lemma}

\begin{proof}
Without loss of generality, consider the case that $i$ is odd. Suppose $|I_i|=t>C(n_i+d'_i)$. Since $I_i$ is a union of active interval(s) that intersect(s) at a single point, there must exist a node $u$ whose active interval is contained within $I_i$ and the length of the active interval of $u$---denoted as $a_u$---is at least $t/2>C(n_i+d'_i)/2$.  We will prove that this event happens with small probability for every node whose active interval is contained within $I_i$, by then a union bound would imply the desired result.

We now prove $\Pr\left[~a_u=t/2 ~\mid~ |I_1|=\ell_1,|I_3|=\ell_3,...,|I_{i-2}|=\ell_{i-2}~\right]\le1/t^{\Omega(1)}$. In the first $t/2$ slots after $u$'s arriving, there are at most $t/C$ busy slots since $t>Cd'_i$. Hence, when $C\geq 4$, during the first $t/2$ slots after $u$'s arriving, there are at least $t/4$ non-jamming slots in which the contention on the second channel is at most $0.5$. Call these slots $G_u$. Notice that $u$ sends with probability at least $(c\log t)/t$ on the second channel in these slots. Therefore, by Lemma~\ref{lem:contention-to-succeess-probability}, for each slot in $G_u$, the probability that $u$ succeeds on channel two (i.e., $u$ sends on channel two and other nodes are idle on channel two) is at least $(c\log t)/t\cdot\Theta(1)$. Since $I_1,I_3,...I_{i-2}$ do not intersect with $I_i$, this probability holds even if it is conditioned on $|I_1|=\ell_1,|I_3|=\ell_3,...,|I_{i-2}|=\ell_{i-2}$. Thus, the probability that node $u$ does not succeed during $G_u$ is at most $\left(1-(c\log t)/t\cdot\Theta(1)\right)^{t/4}={1}/{t^{\Omega(1)}}$, for sufficiently large $c$ and $t$.

Recall that $t>Cn_i$. By the union bound, the probability that there exists a node that lasts for $t/2$ slots before leaving the system is at most $n_i\cdot 1/t^{\Omega(1)}=1/t^{\Omega(1)}$, which proves the lemma.
\end{proof}

We are now ready to state and prove the main technical lemma of this part.

\begin{lemma}\label{lem:dynamic-num-small-contention-slots}
Suppose during the execution of the two-channel algorithm there are $a$ active slots in which the contention on the second channel is at least $0.5$. Then, the total number of active slots is at most $O(a+n+d)$, with high probability in $n+d$.
\end{lemma}

\begin{proof}
According to the definition of complete interval, the total number of active slots can be bounded by $\sum_{i=1}^k|I_i|$. For every $i\in[k]$, let $\ell_i$ be a random variable that equals to $|I_i|$ if $|I_i|>C(n_i+d'_i)$ and equals to $0$ otherwise. (Recall $n_i$ and $d'_i$ are introduced in Lemma~\ref{lem:dynamic-complete-intervals}.) Let $\ell_i=0$ for any $i>k$. According to Lemma~\ref{lem:dynamic-complete-intervals} and Lemma~\ref{lem:cjz21}, we have $\sum_{i=1}^{(n+d)/2}\ell_{2i-1}=O(n+d)$ with high probability in $n+d$; similarly, we also have $\sum_{i=1}^{(n+d)/2}\ell_{2i}=O(n+d)$ with high probability in $n+d$. Thus, $\sum_{i=1}^k|I_i|\leq\sum_{i=1}^k(\ell_i+C(n_i+d'_i))\leq O(n+d)+\sum_{i=1}^kC(n_i+d'_i)=O(n+d+d')$, with high probability in $n+d$. Recall that $d'\leq a+d$, the lemma is proved.
\end{proof}

We conclude this subsection with the proof of Theorem~\ref{thm:dynamic-two-channel}.

\begin{proof}[Proof of Theorem~\ref{thm:dynamic-two-channel}]
Due to Lemma~\ref{lem:dynamic-num-large-contention-slots} and Lemma~\ref{lem:dynamic-num-small-contention-slots}, we know the two-channel algorithm achieves $(n\log{n},d)$-time-cost. That is, the total number of active slots is at most $O(n\log{n}+d)$, w.h.p.\ in $n+d$.

Assume indeed the total number of active slots is $O(n\log{n}+d)$, we now consider energy complexity. Fix a node $u$. If $u$ is a normal node, then $u$ runs \backoff{$1/x$} on the first channel and runs \backoff{$(c\log{x})/x$} on the second channel, until it succeeds. Otherwise, if $u$ is an interference node, then it runs \backoff{$(c\log{x})/x$} until it sees a success. Hence, during $u$'s life cycle, the total contention created by it is at most $O(\log^2(n\log{n}+d))$. Notice that the total contention created by $u$ is also the expected number of times $u$ broadcast during its life cycle. Moreover, $u$'s broadcast decisions are independent between slots. By a Chernoff bound, we can conclude that the number of times $u$ broadcast is $O(\log^2(n\log{n}+d))$, with high probability in $n+d$. Notice that $O(\log^2(n\log{n}+d))$ could also be written as $O(\log^2{n}+\log^2{d})$, thus the theorem is proved.
\end{proof}

\subsection{Conversion to Single-channel Model}

In this subsection, we prove Lemma~\ref{lem:two-to-single} by introducing a method to covert the two-channel algorithm to a corresponding single-channel algorithm. Notice that in the single-channel setting, we can give each slot a global index. If nodes have access to these indices, implementing the two-channel model is easy: each node treats odd slots in the single-channel setting as the first channel in the two-channel model, and treats even slots in the single-channel setting as the second channel in the two-channel model. Nonetheless, we do not assume nodes can access global slot indices. Instead, in the single-channel setting, when a node is injected, it will first run a ``synchronization procedure'' to reach agreement on the parity of slots with the nodes that are already in the system. Similar idea has been used in several previous work~\cite{bender18,bender20,chen21}. More details on this synchronization procedure is presented in the following proof.

\begin{proof}[Proof of Lemma~\ref{lem:two-to-single}]
We first describe our single-channel algorithm for each node $u$. Suppose $u$ joins the system at global slot $t_u$, then $u$ will ``split'' the single channel into two channels: ``odd channel'' containing slots $\{t_u,t_u+2,t_u+4,...\}$, and ``even channel'' containing slots $\{t_u+1,t_u+3,...\}$. For each of the two channels, $u$ will run \backoff{$(c\log x)/x$} on it, starting from slot $t_u$ and $t_u+1$ respectively. Assume later $u$ sees the first successful message transmission in global time slot $T$. At that point, $u$ treats the channel $T$ belongs to as the first channel in the two-channel model, and treats the other channel as the second channel in the two-channel model. Then, $u$ runs the two channel algorithm starting from time slot $T+2$. While running the two-channel algorithm, whenever $u$ hears a success on the second channel in some time slot $T'$ (here $T'$ is a single-channel slot index), it swaps the role of the two channels. Specifically, $u$ treats the current first channel as the second channel and treats the current second channel as the first channel; it does nothing in time slot $T'+1$ (which is now a slot belonging to the second channel) and continues the two-channel algorithm in time slot $T'+2$ (which is now a slot belonging to the first channel). It is straightforward to verify, whenever a success occurs, all nodes that are currently in the system reaches agreement on which slots belong to the first channel (and which slots belong to the second channel).

Now, observe that a node injected into the system (in the single-channel model) can be seen as an interference node injected by the adversary (in the two-channel model). If a node hears a success for the first time (in the single-channel model), then it will start running the two-channel algorithm; this can been seen as the interference node leaving the system and a new normal node is injected (in the two-channel model). At this point, we have finished describing how an arbitrary execution of the single-channel algorithm can be converted into a corresponding execution of the two-channel algorithm, assuming the adversary has the ability to \interference{$(c\log{x})/x$} and can inject $n$ interference nodes in the two-channel model.

Next, we analyze the time complexity. Notice that in the single-channel model, each time a node hearing a success in some slot $T$ that corresponds to the second channel, it will ``waste'' a time slots $T+1$, and continue in slot $T+2$. However, the total number of such wasted slot is at most $n$. Therefore, if the two-channel algorithm achieves $(n\log{n}, d)$-time-cost, then the single-channel algorithm described above has at most $2\cdot O(n\log n+d)+O(n)=O(n\log n+d)$ active slots, implying the single-channel algorithm also achieves $(n\log{n}, d)$-time-cost.

Lastly, we analyze the energy complexity. By the conversion procedure described above, for each node $u$ injected in the single-channel setting, it will first play the role of an interference node, and then the role of a normal node. Recall Definition~\ref{def:throughput-and-energy-with-interference}, in each of the above two parts of the life cycle of $u$, since the two-channel algorithm achieves $O(\log^2{n}+\log^2{d})$-energy-cost, we know with high probability in $n+d$ node $u$'s energy cost is $O(\log^2{n}+\log^2{d})$. Therefore, by Definition~\ref{def:throughput-and-energy}, the single-channel algorithm also achieves $(n\log{n}, d)$-time-cost.
\end{proof}

Finally, we conclude this section by noting that combining Lemma~\ref{lem:two-to-single} and Theorem~\ref{thm:dynamic-two-channel} immediately gives Theorem~\ref{thm:dynamic-single-channel}, which states the time complexity and the energy complexity of our dynamic algorithm in the single-channel setting.

\section{The Static Scenario}\label{sec:static-alg}

In this section, we introduce our algorithm for the scenario where all nodes start execution simultaneously. (That is, Eve activates all nodes in the first slot.) Similar to Section~\ref{sec:dynamic-alg}, we will first give a more through discussion on the intuition of the algorithm, then present the pseudocode, and finally proceed to the analysis.

In the static scenario, all nodes agree on slot indices, so the two-channel model is easy to implement: odd slots correspond to the first channel, which we also refer to as the ``data channel''; and even slots correspond to the second channel, which we also refer to as the ``control channel''.

Assume indeed nodes can access two channels, recall the static algorithm contains two phases, as we have discussed in Section~\ref{sec:overview}. Specifically, in the first phase, nodes uses a variable $\ell$ to control its sending probability: in each slot, send with probability $1/\ell$ on the data channel, and with probability $(c\log\ell)/\ell$ on the control channel. Initially $\ell=1$. In each slot, if a control channel success occurs then $\ell=\ell/2$, effectively doubling the sending probability of the next slot; otherwise simply update $\ell=\ell+1$. This phase one algorithm efficiently maintains the contention of the data channel within a desirable interval, enforcing a total time complexity of $O(n+d)$.

However, the above phase one algorithm could suffer poor energy efficiency when jamming from the adversary is strong (particularly, when $d=\Omega(n\log{n})$). As a result, in the final algorithm, when nodes realize the adversary has jammed a lot, they will switch to a phase two algorithm, which is simply running \backoff{$(c\log{x})/x$} on the control channel. Assuming $n'$ nodes remain when phase two starts, the energy complexity of each such node is $O(\log^2{n'}+\log^2{d})$, and the time complexity of phase two would reach $O(n'\log{n'}+d)$. If we want to enforce the static algorithm's total time complexity to be $O(n+d)$, then we need to ensure $n'\log{n'}=O(d)$. Therefore, the transition from phase one to phase two is crucial: it cannot happen too late, otherwise phase one would incur large energy cost; it also cannot happen too early, otherwise it might be the case that $n'\log n'=\omega(d)$, resulting in sub-optimal total time complexity.

It turns out there is a simple criterion. Upon arrival, each node maintains a variable called $m$ and sets $m=\infty$. During phase one, whenever a control channel success occurs and $\ell$ is halved, nodes update $m=\min\{m,\ell\}$. By the end of each slot $t$, if $m\leq t/\log{t}$, remaining nodes will switch to phase two.

We now explain why the above criterion is the proper transition condition.

\smallskip\underline{\emph{The value of $m$ is an asymptotic upper bound on the number of remaining nodes.}} Recall the contention of the control channel will increase (particularly, double) only if a success occurs on the control channel. Moreover, due to Lemma~\ref{lem:contention-to-succeess-probability}, with high probability in $\ell$ a success will not happen when the contention of the control channel reaches $\Omega(\log\ell)$. Hence, the algorithm restricts the contention of the control channel to be $O(\log\ell)$. Recall that the contention of the two channels differ by a $\Theta(\log \ell)$ factor, thus the contention of the data channel is restricted to be $O(1)$. On the other hand, by definition, the contention of the data channel is $n'\cdot (1/\ell)$ where $n'$ is the number of remaining nodes currently. As a result, we conclude $n'=O(\ell)$ always hold throughout phase one. Recall that $m$ is always the minimum value of $\ell$ up to now and the number of remaining nodes only decrease over time, we conclude $m$ is always an asymptotic upper bound on the number of remaining nodes.

\underline{\emph{Energy complexity of phase one is guaranteed.}} Consider phase one of the algorithm. Assume it takes $T_a$ slots for the first control channel success to occur. Then, for each slot after the first control channel success, $\Theta((\log m)/m)$ is an upper bound on nodes' sending probabilities as by then $m=O(\ell)$ always hold. Recall that for each slot $t$ in phase one, $m=\Omega(t/\log t)$. Moreover, $m$ only decrease overtime. As a result, if phase one ends by the end of slot $T$, then for each slot in $[T_a+1,T]$, each node's sending probability is $O(\log(T/\log{T})/(T/\log{T}))=O((\log^2{T})/T)$. This implies the expected energy cost of each node for phase one is $O(\log^2{T_a}+(T-T_a)\cdot(\log^2{T})/T)=O(\log^2{T})$. Recall that phase one lasts for at most $O(n+d)$ slots (otherwise all nodes would have succeeded), we conclude the energy cost of each node during phase one is $O(\log^2(n+d))=O(\log^2n+\log^2d)$.

\underline{\emph{Time complexity of phase two is guaranteed.}} Due to the above discussion, we have $n'=O(m)=O(T/\log T)$, where $n'$ is the number of nodes entering phase two and $T$ is the duration of phase one. Moreover, $T=O(n+d)$ as phase one algorithm has time complexity $O(n+d)$. Therefore, we conclude $n'\log n'=O(n+d)$, implying the time complexity of phase two is $O(n'\log{n'}+d)=O(n+d)$.

\smallskip In summary, $m=\Omega(t/\log t)$ guarantees the energy complexity of phase one, while $m=O(t/\log t)$ guarantees the time complexity of phase two. Since $m$ is decreasing over time and $t$ is increasing over time, the first time $m\leq t/\log t$ is satisfied is the proper moment to transfer from phase one to phase two.

\subsection{Algorithm Description}\label{sec:static-alg-description}

Our static algorithm is shown below, it resembles the structure we illustrated above. To simplify presentation, it is described in the two-channel model, and can be easily converted to the single-channel setting.

\par\addvspace{.5\baselineskip}
{\centering
\framebox[\textwidth][c]{
\parbox{.95\textwidth}{
\par\addvspace{.3\baselineskip}
\hypertarget{alg:static}{\underline{\emph{Algorithm for each node in the static scenario:}}}
\par\addvspace{.3\baselineskip}
\textbf{Phase I:} Initialize $\ell\gets 1$, $t\gets 1$, $m\gets\infty$. Repeat the following in each slot until $m\le t/\log{t}$, or the node succeeds (and halts).
\begin{enumerate}[nosep]
	\item Send with probability $\min\{1/\ell,1\}$ on the data channel, and with probability $\min\{(c\log\ell)/\ell,1\}$ on the control channel, where $c$ is a sufficiently large constant.
	\item If in this slot the control channel generates a success, then let $\ell\gets\max\{\ell/2,1\}$ and let $m\gets\min\{m,\ell\}$; otherwise let $\ell\leftarrow\ell+1$.
	\item $t\gets t+1$.
\end{enumerate}
\textbf{Phase II:} Run a fresh instance of \backoff{$(c\log{x})/x$} on the control channel, until the node succeeds.
}}}
\par\addvspace{.5\baselineskip}

\subsection{Algorithm Analysis}\label{sec:static-alg-analysis}

We first argue the energy consumption of each node is poly-logarithmic in its total running time.

\begin{lemma}\label{lem:static-energy}
Each node running the static algorithm sends $O(\log^2{t})$ times in the first $t$ slots, w.h.p.\ in $t$.
\end{lemma}

\begin{proof}
Consider a node $u$. Suppose the first control channel success happens in slot $t_a$, then in the first $t_a$ slots $u$ will stay in phase one since $m$ remains $\infty$. Moreover, in the first $t_a$ slots, node $u$ runs a single instance of \backoff{$1/x$} on the data channel, and a single instance of \backoff{$(c\log x)/x$} on the control channel. Hence, in the first $t_a$ slots, in expectation, $u$ sends for at most $O(\log^2{t_a})$ times.

Suppose $u$ ends phase one at slot $t_b$ and the value of $m$ is $m_b$ by the end of $t_b$, then the largest sending probability during interval $[t_a+1,t_b]$ is $(c\log{m_b})/{m_b}$, since $m$ is non-increasing. We now argue $m_b=\Omega(t_b/\log{t_b})$. Since slot $t_b$ is the first time that condition $m\le t/\log t$ is satisfied, we know by the end of slot $t_b-1$ the value of $m$---denoted by $m_{t_b-1}$---must satisfy $m_{t_b-1}>(t_b-1)/\log{(t_b-1)}=\Omega(t_b/\log{t_b})$. If $m_b=m_{t_b-1}$ then we have proved $m_b=\Omega(t_b/\log{t_b})$, so assume $m_b\neq m_{t_b-1}$.
Suppose the value of $m$ becomes $m_{t_b-1}$ in some slot $t_{b'}\leq t_b-1$; further suppose $\ell_{t_{b'}}$ and $\ell_{t_{b}}$ are the values of $\ell$ by the end of slot $t_{b'}$ and slot $t_b$ respectively. Then by algorithm description we know $2\ell_{t_b}\geq \ell_{t_{b'}}$, which further implies $2m_{b}\geq m_{t_b-1}$. Again, we have $m_b\geq m_{t_b-1}/2=\Omega(t_b/\log{t_b})$. At this point, we have proved $m_b=\Omega(t_b/\log{t_b})$, and the largest sending probability of $u$ during interval $[t_a+1,t_b]$ is $(c\log{m_b})/{m_b}$. Hence, during interval $[t_a+1,t_b]$, in expectation, $u$ sends for at most $t_b\cdot(c\log{m_b})/{m_b}=O(\log^2{t_b})$ times.

Lastly, consider phase two, in which node $u$ runs a simple \backoff{$(c\log x)/x$}. Suppose $u$ stops at some slot $t\geq t_b$, we know in expectation, during phase two, $u$ sends for at most $O(\log^2{(t-t_b)})$ times.

Since in each slot during the entire life cycle of $u$ it makes broadcasting decisions independently, and since $t_b\leq t_b\leq t$, by the Chernoff bound, we can conclude node $u$ will send at most $O(\log^2 t)$ times during the first $t$ slots, with high probability in $t$.
\end{proof}

In the following, we argue the time complexity of our static algorithm, and we do so by proving a bound for phase one and phase two separately. We will first focus on phase two since it is easier to analyze.

\begin{lemma}\label{lem:static-phase2-time}
Denote the duration of phase two as $T$, denote the number of nodes that start phase two as $n'$. Then for any $t\ge c^2(n'\log n'+d)$, we have $\Pr[T=t]\le 1/t^{\Omega(1)}$.
\end{lemma}

\begin{proof}
After $10cn'\log n'$ slots into phase two, for each following slot, the contention on the control channel will be at most $0.5$. Moreover, each node sends with probability at least $(c\log t)/t$ on the control channel in the first $t$ slots in phase two. Thus, assuming $t\ge c^2(n'\log n'+d)$, for sufficiently large $c$, for any slot in the interval $[10cn'\log n',t]$, each node $u$ has probability at least $(2^{-2\cdot 0.5})\cdot(c\log t)/t=(0.5c\log t)/t$ to broadcast its message alone. (Due to Lemma~\ref{lem:contention-to-succeess-probability}, in each such slot, the probability that all nodes except $u$ do not send is at least $2^{-2\cdot 0.5}$.) Moreover, in the interval $[10cn'\log n',t]$, for sufficiently large $c$, Eve can jam at most $t/4$ slots. Therefore, for sufficiently large $c$, the interval $[10cn'\log n',t]$ contains at least $t/2$ slots such that in each such slot, a node can successfully send its message with probability at least $(0.5c\log t)/t$. Therefore, apply a Chernoff bound, we know for each node $u$ that joins phase two, if $c$ is sufficiently large and $t\ge c^2(n'\log n'+d)$, then $u$ will succeed and halt within the first $t$ slots of phase two, with high probability in $t$. Take a union bound over the $n'\leq t$ nodes that join phase two, the lemma is proved.
\end{proof}

Next, we analyze the time complexity of phase one. To that end, we define \emph{complete intervals}. Specifically, a complete interval during phase one is the time interval between two (consecutive) successes which both occur on the control channel. Clearly, phase one can be partitioned into a set of complete intervals.

The following lemma shows the first complete interval will not be too long.

\begin{lemma}\label{lem:static-phase-one-first-success}
In phase one, the first control channel success happens in $O(n+d)$ slots, w.h.p.\ in $n+d$.
\end{lemma}

\begin{proof}
Before the first control success, all nodes are running \backoff{$(c\log x/x)$} on the control channel. Suppose $t=c^2(n+d)$, consider interval $[t/2,t]$. If no control success occurs in the first $t$ slots, then for each slot in interval $[t/2,t]$, the contention on the control channel is at least $(c\log t)/t$ and at most $n\cdot(c\log(t/2)/(t/2))\leq 2cn(\log t)/t\le (2\log t)/c$. Thus, according to Lemma~\ref{lem:contention-to-succeess-probability}, for sufficiently large $c$, in each of these slots, the probability that some node broadcasts alone on the control channel is at least $(c\log t)/(4t)$. Notice that in interval $[1,t]$, the control channel is jammed by Eve for at most $t/c^2$ times. So in interval $[t/2,t]$, for sufficiently large $c$, there are at least $t/4$ slots in which a success will occur on the control channel with probability at least $(c\log t)/(4t)$. Therefore, a success will occur in interval $[t/2,t]$ with high probability in $t$. Recall that $t=c^2(n+d)$, the lemma is proved.
\end{proof}

Then, we bound the duration of the complete intervals after the first control channel success, since these complete intervals can be amortized by the number of \emph{congest slots} defined as below.

\begin{definition}\label{def:static-congest-slot}
In phase one, we say a slot is a \emph{congest slot} if it is jammed by the adversary (on either channel), or the contention of the data channel in this slot is at least $1/c^2$.
\end{definition}

The following lemma shows the total number of congest slots cannot be too large.

\begin{lemma}\label{lem:static-cogest-slot-num}
In phase one, the number of congest slots in the first $(n+d)^2$ slots is $O(n+d)$, w.h.p.\ in $n+d$.
\end{lemma}

\begin{proof}
We first introduce the following claim which is helpful for proving the lemma.

\begin{claim}\label{claim:static-claim}
In phase one, the following event happens w.h.p.\ in $n+d$: for each slot after the first control channel success and before time slot $(n+d)^2$, the contention on the control channel is $O(\log(n+d))$.
\end{claim}
\begin{proof}
According to Lemma~\ref{lem:contention-to-succeess-probability}, if in a slot the contention of the control channel is at least $\Omega(\log(n+d))$, then the probability that the control channel generates a success is at most $1/(n+d)^{\Omega(1)}$. Take a union bound over $(n+d)^2$ slots, we know the following event happens with high probability in $n+d$: for each slot in the first $(n+d)^2$ slots, if a success occurs on the control channel in that slot, then the contention on the control channel in that slot is $O(\log(n+d))$. Notice that when the first control channel success occurs during phase one, the contention of the control channel is $O(\log{(n+d)})$, with high probability in $n+d$. Moreover, after the first control channel success, the only possibility to raise the contention of the control channel is a control channel success. Thus, the claim is proved.
\end{proof}

We now proceed to prove the lemma. Specifically, we shall prove that in phase one, the number of congest slots after time slot $C(n+d)$ is $O(n+d)$, for some sufficiently large constant $C$.

We claim, with high probability in $n+d$, for any phase one slot after $C(n+d)$ but before $(n+d)^2$, the contention of the data channel in that slot is $O(1)$. To see this, notice that according to Lemma~\ref{lem:static-phase-one-first-success}, with high probability in $n+d$, by slot $C(n+d)$, the first control channel success has occurred. Assume indeed the first control channel success has occurred by slot $C(n+d)$, then due to Claim~\ref{claim:static-claim}, in phase one, with high probability in $n+d$, for each slot after slot $C(n+d)$ but before $(n+d)^2$, the contention of the control channel in that slot is $O(\log(n+d))$. On the other hand, for any phase one slot $t\geq C(n+d)$, let $\ell_t$ denote the value of $\ell$ at the end of that slot. If $\ell_t\leq t/\log{t}$, then by algorithm description, phase one will end by the end of slot $t$. So for each phase one slot $t\in [C(n+d)+1,(n+d)^2]$, $\log{\ell_t}=\Omega(\log{(n+d)})$. Since for each phase one slot $t\in [C(n+d)+1,(n+d)^2]$, the contention of the data channel is $1/\Theta(\log(\ell_t))$ fraction of the contention of the control channel, and recall that the contention of the control channel in each such slot is $O(\log(n+d))$, we conclude that the contention of the data channel in each such slot is $O(1)$.

Now, for each phase one slot in interval $[C(n+d)+1,(n+d)^2]$, if it is not jammed by Eve but a congest slot, then by definition in this slot the contention on the data channel is at least $1/c^2$. Recall we have also shown above in this slot the contention on the data channel is $O(1)$, with high probability in $n+d$. So by Lemma~\ref{lem:contention-to-succeess-probability}, in this slot, a success will occur on the data channel with some constant probability. Therefore, in phase one, after $\Theta(n+d)$ such slots (i.e., congest slots without jamming), $n$ data channel successes must have occurred and all nodes must have succeeded, with high probability in $n+d$. Since the jamming from Eve can create at most $d$ additional congest slots, the lemma is proved.
\end{proof}

At this point, we are ready to formally define complete interval and bound its length. Notice that in the following lemma, besides congest slots, we also use the change in the value of $\ell$ to amortize interval length.

\begin{lemma}\label{lem:static-phase1-time}
Suppose $k$ control channel successes happen during phase one, and the $i$-th success happens in slot $t_{i+1}$. Define $t_1=0$, and define $t_{k+1}$ be the last slot of phase one. For every $i\in[k]$, define complete interval $I_i=[t_i+1,t_{i+1}]$, define $d'_i$ to be the number of congest slots during $I_i$, and define $\Delta\ell_i$ to be the value of $\ell$ at the beginning of slot $t_{i}+1$ minus the value of $\ell$ at the end of slot $t_{i+1}$. Then for any positive integer $t>C(d'_i+\Delta\ell_i)$, where $C$ is a sufficiently large constant, we have $\Pr\left[~|I_i|=t ~\mid~ |I_1|=T_1,|I_2|=T_2,...,|I_{i-1}|=T_{i-1}~\right]\le1/t^{\Omega(1)}$, for any fixed $T_1,T_2,...,T_{i-1}$.
\end{lemma}

\begin{proof}
Fix an interval $I_i$ and some $t>C(d'_i+\Delta\ell_i)$, suppose $|I_i|=t$. Suppose at the beginning of slot $t_{i}+1$, the value of $\ell$ is $\ell_{t_i+1}$. Notice that we assume $t>C(d'_i+\Delta\ell_i)$, and by definition $\Delta\ell_i=\ell_{t_i+1}-(\ell_{t_i+1}+t-1)/2>(\ell_{t_i+1}-t)/2$, thus $d_i'\le 0.75t$ when $C\ge 4$. Therefore, there will be at least $0.25t$ slots satisfying: (a) not jammed by the adversary; and (b) the contention on data channel is at most $1/c^2$. Let $G$ denote this set of slots. For each slot in $G$, we can further conclude the contention of the control channel is at most $(1/c^2)\cdot(c\log{(\ell_{t_i+1}+t)})=\log(\ell_{t_i+1}+t)/c$. Notice that we assume $t>C\Delta\ell_i$ and $\Delta\ell_i=\ell_{t_i+1}-(\ell_{t_i+1}+t-1)/2>(\ell_{t_i+1}-t)/2$, so we know $\ell_{t_i+1}<3t$. Therefore, for each slot in $G$, the contention on the control channel is at most $\log(\ell_{t_i+1}+t)/c<\log(4t)/c$, and at least $c\log(\ell_{t_i+1}+t-1)/(\ell_{t_i+1}+t-1)>c\log(\ell_{t_i+1}+t)/(\ell_{t_i+1}+t)>c\log(4t)/(4t)$. Therefore, by Lemma~\ref{lem:contention-to-succeess-probability}, we know for sufficiently large $c$, for each slot in $G$, the probability that the control channel will generate a success is at least $c\log(4t)/(16t)$. Recall that $|G|>0.25t$, we know with probability at least $1-(1-c\log(4t)/(16t))^{t/2}\geq 1-1/t^{\Omega(1)}$, at least one success will occur on the control channel in $G$.
\end{proof}

We are now ready to prove the guarantees enforced by our static algorithm.

\begin{proof}[Proof of Theorem~\ref{thm:static-single-channel}]
We first bound the length of phase one. Assume $k$ control channels successes occur during phase one, so phase one can be divided into $k$ complete intervals. For every $i\in[k]$, let $T_i$ be a random variable such that $T_i=|I_i|$ if $|I_i|>C(d'_i+\Delta\ell_i)$, otherwise $T_i=0$. (Recall $I_i,C,d'_i,\Delta\ell_i$ are defined in Lemma~\ref{lem:static-phase1-time}.) For every $i\in[k+1,n+d]$, define $T_i=0$. Therefore, the duration of phase one is $\sum_{i=1}^{k}|I_i|\leq\sum_{i=1}^{k}(T_i+C(d'_i+\Delta\ell_i))=\sum_{i=1}^{k}T_i+C\cdot\sum_{i=1}^{k}(d'_i+\Delta\ell_i)=\sum_{i=1}^{n+d}T_i+C\cdot\sum_{i=1}^{k}(d'_i+\Delta\ell_i)$. Due to Lemma~\ref{lem:static-phase1-time} and Lemma~\ref{lem:cjz21}, we know $\sum_{i=1}^{n+d}T_i=O(n+d)$, with high probability in $n+d$. Due to the analysis in the proof of Lemma~\ref{lem:static-cogest-slot-num}, we know $\sum_{i=1}^{k}d'_i=O(n+d)$, with high probability in $n+d$. Lastly, notice that $\Delta\ell_i=\ell_{t_i+1}-\ell_{t_{i+1}}/2$ and $\ell_{t_{i+1}+1}=\ell_{t_{i+1}}/2$, we know $\sum_{i=1}^{k}\Delta\ell_i=\ell_{t_1+1}-\ell_{t_{k+1}+1}/2\leq\ell_{t_1+1}=1$. Therefore, we can conclude, the time complexity of phase one is at most $O(n+d)+C\cdot O(n+d)+C\cdot 1=O(n+d)$, with high probability in $n+d$.

Next, we bound the length of phase two, and we consider two complement cases depending on the relationship between $n$ and $d$. The first case is $d=\Omega(n\log n)$. According to Lemma~\ref{lem:static-phase2-time}, in such case, the second phase runs for $O(n\log n+d)$ slots, with high probability in $n+d$. The other case is $d=O(n\log n)$, and its analysis is more involved. (Notice, when $d=O(n\log n)$, if an event happens with high probability in $n$, then it also happens with high probability in $n+d$.)

Assume $d=O(n\log n)$, we first prove the claim that in phase one the first control channel success happens after time slot $n$, with high probability in $n$. To see this, notice that in phase one, for each slot before slot $n$ and before the first control channel success, the contention on the control channel is at least $n\cdot (c\log n)/n\geq c\log{n}$. Hence, due to Lemma~\ref{lem:contention-to-succeess-probability}, with high probability in $n$, a control channel success will not occur in this slot. Take a union bound over the first $n$ slots in phase one, the claim is proved.

Then we argue, in phase one, for any slot after the first control channel success and before slot $(n+d)^2$, the contention of the data channel in that slot is $O(1)$, with high probability in $n$. To see this, notice that due to Claim~\ref{claim:static-claim} and the assumption $d=O(n\log n)$, in phase one, for any slot after the first control channel success and before slot $(n+d)^2$, the contention of the control channel in that slot is $O(\log{n})$. On the other hand, in phase one, if the first control channel success happens after slot $n$ (this happens with high probability in $n$ due to the discussion in the last paragraph), then by algorithm description, $\ell\geq n/\log n$ always holds after the first control channel success. (Otherwise, phase one will end.) Lastly, notice that in phase one, for each slot after the first control channel success, the contention on the data channel and the control channel differ by a factor of $c\log{\ell}$, which is at least $\Omega(\log{n})$ since we have just shown $\ell\geq n/\log n$. At this point, we can conclude, in phase one, for any slot after the first control channel success and before slot $(n+d)^2$, the contention of the data channel in that slot is $O(1)$, with high probability in $n$.

Now, suppose phase one ends by the end of slot $T$ and $n'$ nodes remain. Suppose the value of $\ell$ at the beginning of $T$ is $\ell_T$. By algorithm description, we know $\ell_T/2\leq T/\log{T}$. Since in slot $T$ each node sends on the data channel with probability $1/\ell_T$ and we have shown above with high probability in $n$ the contention on the data channel is $O(1)$ in that slot, we know $n'=O(1)/(1/\ell_T)=O(\ell_T)=O( T/\log{T})$ with high probability in $n$. Recall we have already shown the time complexity of phase one---which is $T$---is $O(n+d)$, with high probability in $n+d$. Hence, when $d=O(n\log n)$, by Lemma~\ref{lem:static-phase2-time}, the time complexity of phase two is $O(n'\log n'+d)=O(n+d)$, with high probability in $n$.

At this point, we have proved the total time complexity of the static algorithm is $O(n+d)$, with high probability in $n+d$. Finally, recall Lemma~\ref{lem:static-energy}, we can also conclude the energy complexity of each node is $O(\log^2{n}+\log^2{d})$, with high probability in $n+d$.
\end{proof}

\section{Lower Bounds}\label{sec:lower-bounds}

In this part, we prove several complexity lower bounds for the contention resolution problem, when collision detection is not available and external interference is present (i.e., jamming). These bounds show both our dynamic algorithm and static algorithm achieve optimal time complexity. As for energy complexity, the dynamic algorithm also matches the lower bound, whereas the static algorithm misses the lower bound by a poly-logarithmic factor. In particular, in the static scenario, the lower bound is $\Omega(\log\log{n}+\log^2{d})$, while our algorithm incurs a per-node energy cost of $O(\log^2{n}+\log^2{d})$.

Before proving the lower bounds, we first introduce the following key technical lemma. Recall the discussion in Section~\ref{sec:overview}, this lemma helps us to connect the time complexity and the energy complexity of a contention resolution algorithm.

\begin{lemma}\label{lem:lower-bound-time-energy}
For any function $f:\mathbb{N}^+\rightarrow\mathbb{R}^+$, if algorithm $\mathcal{A}$ achieves $(f(n),g(d))$-time-cost in the static case with $g(d)=d$, and if algorithm $\mathcal{A}$ does not observe any success in the first $t$ slots, then in expectation algorithm $\mathcal{A}$ sends $\Omega(\log^2 t)$ times in the first $t$ slots.
\end{lemma}

\begin{proof}
Consider the case one single node $u$ runs $\mathcal{A}$. Since $\mathcal{A}$ achieves $(f(n),d)$-time cost, there must exist a constant $C$ such that the number of active slots---i.e., the time required for the node to succeed---is at most $C(f(1)+d)$, with high probability in $d$. Here, $d$ is the number of slots jammed by the adversary during the execution. Now, consider the first $t$ slots of the execution, assume Eve uses the following jamming strategy. She jams all of the first $t/(4C)$ slots. Moreover, Eve further jams another $t/(4C)$ slots, chosen uniformly at random from interval $(t/(4C),t]$. In this scenario, by the end of slot $t$, we have $C(f(1)+d)= Cf(1)+t/2<t$ for any $t>2Cf(1)$. Therefore, there is at least one success in the first $t$ slots, with high probability in $d$. Since $d=\Theta(t)$, this claim also holds with high probability in $t$.

Denote $k(t)$ as the number of times node $u$ broadcasts in interval $(t/(4C),t]$. Since in interval $(t/(4C),t]$ the adversary randomly chooses $t/(4C)$ slots to jam, the probability that no success occurs in these slots is at least $\sum_{k\le\frac{t}{8C}}\frac{\Pr[k(t)=k]}{(8C)^k}$. Since there must exist a success in these slots with probability at least $1-1/t$, we have $\sum_{k\le\frac{t}{8C}}\frac{\Pr[k(t)=k]}{(8C)^k}\le\frac{1}{t}$. Since  $\frac{t}{8C}\ge\log_{8C}\frac{t}{2}$ for sufficiently large $t$, we have:
$$\sum_{k\le\frac{t}{8C}}\frac{\Pr[k(t)=k]}{(8C)^k}
\geq \sum_{k\le\log_{8C}\frac{t}{2}}\frac{\Pr[k(t)=k]}{(8C)^{\log_{8C}t/2}} \geq \sum_{k\le\log_{8C}\frac{t}{2}}\frac{\Pr[k(t)=k]}{t/2} = \frac{\Pr\left[k(t)\le\log_{8C}\left(t/2\right)\right]}{t/2}$$
If $\Pr\left[k(t)\le\log_{8C}\frac{t}{2}\right]>1/2$, then $\sum_{k\le\frac{t}{8C}}\frac{\Pr[k(t)=k]}{(8C)^k}>\frac{1/2}{t/2}=\frac{1}{t}$, which is a contradiction. Therefore, $\Pr\left[k(t)>\log_{8C}{\frac{t}{2}}\right]\ge 1-1/2=1/2$, implying $\mathbb{E}[k(t)]\ge\frac{1}{2}\log_{8C}\frac{t}{2}$.

By an argument similar as above, we can show the expected number of times node $u$ broadcasts in interval $\left(\frac{t}{(4C)^{i+1}},\frac{t}{(4C)^i}\right]$ is at least $\frac{1}{2}\log_{8C}\frac{t/2}{(4C)^i}$, for any $1\leq i\leq l$ where $l=\log_{4C}t$. Therefore, the total expected number of times $u$ broadcasts in the first $t$ slots is at least:
$$\sum_{0\le i\le l}\frac{1}{2}\log_{8C}\frac{t/2}{(4C)^i}\ge\sum_{0\le i\le l/2}\frac{1}{2}\log_{8C}\frac{t/2}{(4C)^{l/2}} =\Omega\left(\log^2t\right)$$
\end{proof}

Since an algorithm with $(f(n),g(d))$-time cost in the dynamic case is also an algorithm with $(f(n),g(d))$-time cost in the static case, the above lemma also holds for dynamic algorithms.

We are now ready to prove the lower bounds, and we start by considering the dynamic case.

\begin{proof}[Proof of Theorem~\ref{thm:lower-bound-dynamic}]
$f_t(n)=\Omega(n\log n)$ is implied by Theorem 1.3 of~\cite{chen21}. In the reminder of the proof, we focus on the energy complexity.

We first prove $g_e(d)=\Omega(\log^2d)$. Consider a simple adversary strategy that injects one node in the first slot and jams the first $d$ slots, then there are no successes in the first $d$ slots. According to Lemma~\ref{lem:lower-bound-time-energy}, the injected node will send in expectation $\Omega(\log^2d)$ times in those slots. Suppose $g_e(d)=o(\log^2d)$, then with high probability in $d$, the node sends $o(\log^2d)$ times in the first $d$ slots, this leads to an expectation of at most $o(\log^2d)+d\cdot 1/d^{\Omega(1)}=o(\log^2d)$ times of sending in the first $d$ slots, which is a contradiction.

Then we prove $f_e(n)=\Omega(\log^2n)$. Consider the adversary strategy that injects $\sqrt{n}$ nodes in each of the first $\sqrt{n}$ slots. Without loss of generality, assume the algorithm instructs each node to send in the first slot (after arriving) with probability $x_1>0$. Then, for each of the first $\sqrt{n}$ slots, the contention is at least $x_1\sqrt{n}$. Since the algorithm does not know the value of $n$, the value of $x_1$ cannot depend on the value of $n$; this implies $x_1\sqrt{n}\geq n^{0.4}$ when $n$ is sufficiently large. Therefore, due to Lemma~\ref{lem:contention-to-succeess-probability} and the union bound, for sufficiently large $n$, in the first $\sqrt{n}$ slots, with high probability in $n$ no success will occur. Consider a node $u$ that is injected in the first slot. Due to Lemma~\ref{lem:lower-bound-time-energy}, we can further conclude, in the first $\sqrt{n}$ slots, in expectation $u$ sends at least $\Omega(\log^2\sqrt{n})=\Omega(\log^2{n})$ times. Now, if $f_e(n)=o(\log^2{n})$, then with high probability in $n$, node $u$ sends $o(\log^2n)$ times in the first $\sqrt{n}$ slots; this leads to an expectation of at most $o(\log^2n)+\sqrt{n}\cdot 1/n^{\Omega(1)}=o(\log^2{n})$ times of sending in the first $\sqrt{n}$ slots, which is a contradiction.
\end{proof}

Next, we consider the static case. In such scenario, the $\Omega(n+d)$ time complexity bound is trivial, since in each slot there is at most one success (thus giving the $\Omega(n)$ part) and the adversary can jam $d$ slots to block any success (thus giving the $\Omega(d)$ part). Therefore, we focus on proving the energy complexity.

\begin{proof}[Proof of Theorem~\ref{thm:lower-bound-static}]
Notice the proof for $g_e(d)=\Omega(\log^2d)$ in the proof of Theorem~\ref{thm:lower-bound-dynamic} can be carried over to the static case without any modification, so we only need to show $f_e(n)=\Omega(\log\log{n})$.

Suppose the algorithm sends with probability $x_j$ in the $j$-th slot before hearing any successes. Fix some sufficiently large integer $t$, consider the first $t$ slots, denote the sending probabilities in these slots as $x_1,...,x_t$. Now we define $t+1$ real intervals: for every integer $i\in[0,t]$, let $I_i=[t^{-(i+1)c},t^{-ic})$. Notice, these intervals together constitutes a partition of the real interval $[t^{-(t+1)c},1)$. Due to the pigeonhole principle, we know there must exist some $\hat{i}\in[0,t]$ such that: for all $j\in[t]$, it holds $x_j\notin I_{\hat{i}}$. Now consider the adversary strategy that injects $t^{\hat{i}c+c/2}$ nodes in the first slot. Since $x_j\notin I_{\hat{i}}$ holds for all $j\in[t]$, we know in each of the first $t$ slots, assuming all $t^{\hat{i}c+c/2}$ nodes are still active in that slot, the contention is either at least $t^{-\hat{i}c}\cdot t^{\hat{i}c+c/2}=t^{c/2}$, or at most $t^{-(\hat{i}+1)c}\cdot t^{\hat{i}c+c/2}=t^{-c/2}$. Therefore, by Lemma~\ref{lem:contention-to-succeess-probability}, in each of the first $t$ slots, assuming all $t^{\hat{i}c+c/2}$ nodes are still active in that slot, the probability that the slot generates a success is at most $\max\{(t^{c/2})\cdot e^{-t^{c/2}+1},(t^{-c/2})\cdot e^{-t^{-c/2}+1}\}$, which is $1/t^{\Omega(1)}$ for sufficiently large $c$. That is, in each of the first $t$ slots, assuming all $t^{\hat{i}c+c/2}$ nodes are still active, the slot will not generate a success, with high probability in $t$. By an induction argument and a union bound, we know if Eve injects $n=t^{\hat{i}c+c/2}$ nodes, then no success will occur in the first $t$ slots, with high probability in $t$. Consider a node $u$ that is injected. Due to Lemma~\ref{lem:lower-bound-time-energy}, we can further conclude, in the first $t$ slots, in expectation $u$ sends at least $\Omega(\log^2t)$ times. Now, if $f_e(n)=o(\log\log n)$, then the total energy consumption of $u$ is $o(\log\log{(t^{tc+c/2})})=o(\log{t})$, with high probability in $n$. Since $n\geq t$ when $c\geq 2$, the energy consumption of $u$ in the first $t$ slots is $o(\log{t})$, with high probability in $t$. This leads to an expectation of at most $o(\log{t})+t\cdot 1/t^{\Omega(1)}=o(\log{t})$ times of sending in the first $t$ slots, which is a contradiction.
\end{proof}

Finally, we note that our lower bounds are strong in the sense that they hold even for an oblivious adversary (i.e., an offline adversary). By contrast, our algorithms can tolerate an adaptive adversary. Moreover, since the lower bounds only consider the time and energy required to generate the first success, it is likely that they also hold for the leader election problem in similar models.

\section{Future Work}\label{sec:future-work}

In this last section, we discuss some directions for potential future work.

\smallskip\underline{\emph{Optimal energy dependency on $n$ in the static scenario.}} For the static scenario, it is still not clear whether there exists an algorithm that is able to achieve $o(\log^2{n})+O(\log^2{d})$ energy complexity, when jamming is present. On the other hand, it is also not clear how to improve the $\Omega(\log\log n)$ lower bound. Roughly speaking, the $\Omega(\log\log n)$ lower bound is obtained by first showing $\Omega(\log n)$ time is required to generate the first success, and then apply Lemma~\ref{lem:lower-bound-time-energy}. But there actually exists an algorithm which will likely to generate the first success in $O(\log n)$ time, at least when jamming is not present. (For example, each node sends with probability $1/2^i$ in the $i$-th slot.) Therefore, to understand the exact energy complexity in the static scenario, we need either a more clever algorithm, or some new lower bound technique.

\smallskip\underline{\emph{Energy complexity for various jamming intensity.}} As \cite{chen21} has shown, without collision detection, for different jamming intensities, the optimal throughput are different. In this paper, we consider one interesting case where Eve is allowed to jam constant fraction of all slots. (This is the strongest level of jamming intensity.) In \cite{chang19leader}, the authors show that when the severity of jamming goes down to $0$ (i.e., no jamming), the energy complexity lower bound for generating the first success goes down to $\Omega(\log^*n)$. Therefore, examine how the energy complexity evolves as the jamming severity changes is an interesting and natural problem.

\smallskip\underline{\emph{Cost of channel feedback.}} In this paper, we assume only sending messages incurs energy cost, while obtaining channel feedback (i.e., ``listening'') comes for free. Though this is quite standard in the literature and consistent with many real world applications of contention resolution, in the context of radio networks, this assumption often is no longer valid. There are results counting both ``send'' and ``listen'' as access attempts (see, e.g., \cite{bender16,chang19leader}), but it seems the scenario where collision detection is not available and jamming is present has not been considered yet.

\appendix
\section*{Appendix}

\begin{proof}[Proof of Lemma~\ref{lem:contention-to-succeess-probability}]
Since all $X_i$ are 0-1 random variables, for $\sum_{i=1}^{n}X_i$ to be one, it must be the case that $X_j=1$ for some $j\in[n]$ and all other $X_i$ are zeroes. Therefore, $\Pr[(\sum_{i=1}^{n}X_i)=1]=\sum_{i=1}^{n}(p_i\cdot\prod_{j\in[n]\setminus\{i\}}(1-p_j))\geq\sum_{i=1}^{n}(p_i\cdot\prod_{j\in[n]}(1-p_j))=p\cdot\prod_{i=1}^{n}(1-p_i)$. Since $1-p_i\geq 4^{-p_i}$ when $p_i\in[0,1/2]$, we have $\Pr[(\sum_{i=1}^{n}X_i)=1]\geq p\cdot\prod_{i=1}^{n}4^{-p_i}=p\cdot4^{-p}$. Since $4^{-p}\geq 1/4$ when $p\leq 1$, we conclude $p\cdot 4^{-p}\geq \min(4^{-p},p/4)$.

Next, we upper bound $\Pr[(\sum_{i=1}^{n}X_i)=1]$. Notice that $\Pr[(\sum_{i=1}^{n}X_i)=1]=\sum_{i=1}^{n}(p_i\cdot\prod_{j\in[n]\setminus\{i\}}(1-p_j))\leq\sum_{i=1}^{n}(p_i\cdot\prod_{j\in[n]\setminus\{i\}}e^{-p_j})=\sum_{i=1}^{n}(p_i\cdot e^{-\sum_{j\in[n]\setminus\{i\}}p_j})$. Since $p_i\in[0,1]$, we know $\sum_{i=1}^{n}(p_i\cdot e^{-\sum_{j\in[n]\setminus\{i\}}p_j})\leq \sum_{i=1}^{n}(p_i\cdot e^{-\sum_{j\in[n]}p_j+1})=\sum_{i=1}^{n}(p_i\cdot e^{-p+1})=p\cdot e^{-p+1}$.

Lastly, when $p_i\in[0,1/2]$ for all $i\in[n]$, we have $\Pr[(\sum_{i=1}^{n}X_i)=0]=\prod_{i=1}^{n}\Pr[X_i=0]=\prod_{i=1}^{n}(1-p_i)\geq\prod_{i=1}^{n}2^{-2p_i}=2^{-2p}$.
\end{proof}

\newcommand{\etalchar}[1]{$^{#1}$}

\end{document}